\documentclass[11pt]{article}
\usepackage[utf8]{inputenc}
\usepackage{mathtools}
\usepackage{comment}
\usepackage{amsmath}
\usepackage{amsfonts,amsthm,boxedminipage,color,url,thmtools}
\usepackage{fullpage}
\usepackage{soul}
\usepackage{amssymb}
\usepackage[numbers]{natbib}
\usepackage[pdfstartview=FitH,colorlinks,linkcolor=blue,filecolor=blue,citecolor=blue,urlcolor=blue]{hyperref} 
\usepackage[labelfont=bf]{caption}
\usepackage{aliascnt,cleveref}
\usepackage{authblk}
\usepackage{accents}
\usepackage{tikz}
\usetikzlibrary{fit,arrows.meta}
\usepackage{pgfplots}
\pgfplotsset{width=10cm,compat=1.9}
\usepgfplotslibrary{external}
\usepackage{cleveref}
\crefname{algocf}{algorithm}{algorithms}
\Crefname{algocf}{Algorithm}{Algorithms}

\usepackage[linesnumbered,ruled,vlined]{algorithm2e}
\DontPrintSemicolon
\SetKwInput{KwInput}{Input}
\SetKwInput{KwOutput}{Output}

\newtheorem{theorem}{Theorem}
\newtheorem{lemma}{Lemma}

\newtheorem{claim}{Claim}
\newtheorem{observation}{Observation}
\newtheorem{proposition}{Proposition}
\newtheorem{definition}{Definition}

\newtheorem{remark}{Remark}

\newcommand{\eps}{\varepsilon}
\newcommand{\reals}{\mathbb{R}}

\newcommand{\E}{\mathbb{E}}

\newcommand{\M}{\mathcal{M}}
\newcommand{\marg}[2]{f_{#1}(#2)}
\newcommand{\contract}{\boldsymbol{\alpha}}

\newcommand{\safe}{safe\xspace}
\newcommand{\Safe}{Safe\xspace}

\DeclareMathOperator*{\argmax}{arg\,max}

\newcommand{\Var}{\mathrm{Var}}

\title{Online Multi-Agent Contracts\thanks{This project has been partially funded by the European Research Council (ERC) under the European Union's Horizon 2020 program (grant agreement No.~866132), by the European Union's Horizon Europe Program (grant agreement No.~101170373), by an Amazon Research Award, by the Israel Science Foundation Breakthrough Program (grant No.~2600/24), and by a grant from TAU Center for AI and Data Science (TAD), and by the NSF-BSF (grant number 2020788).}}

\date{\today}

\author{
Paul Dütting\thanks{Google Research, Zurich, Switzerland. Email: \texttt{duetting@google.com}} \qquad 
Michal Feldman\thanks{Tel Aviv University and Microsoft ILDC, Israel. Email: \texttt{mfeldman@tauex.tau.ac.il}} \qquad 
Yoav Gal-Tzur\thanks{Tel Aviv University, Israel. Email: \texttt{yoavgaltzur@mail.tau.ac.il}} \qquad 
Thomas Kesselheim\thanks{University of Bonn, Bonn, Germany. Email: \texttt{thomas.kesselheim@uni-bonn.de}}
}

\begin{document}

\maketitle

\begin{abstract}
    We introduce and study an online variant of the multi-agent contract model. In our model, agents arrive one-by-one and are active with a certain probability. Upon arrival of agent $i$, the principal offers a linear contract $\alpha_i$, specifying the fraction of the principal's reward transferred to agent $i$. Agents can either exert effort or not, incurring a cost if they do. The set of agents that exert effort determines the principal's expected reward through a reward function $f$. After all agents have arrived, the agents form a (pure) Nash equilibrium. As our main result we design an $O(1)$-competitive policy for submodular rewards, compared to the offline optimum. We also show that this result is tight in two ways. First, if we require that agents make decisions on the spot, then for submodular rewards 
    any policy is $\Omega((\log n)/(\log \log n)^2)$-competitive.
    Second, for the broader class of XOS (a.k.a., fractionally subadditive) rewards, any online policy is $\Omega((\log \log n)/(\log \log \log n))$-competitive.
    The latter result reveals a surprising separation between submodular and XOS rewards: unlike related settings such as offline contract design and prophet inequalities, where constant-factor guarantees for submodular rewards  extend to XOS rewards, the online contract setting separates the two classes.
\end{abstract}

\section{Introduction}

There are many situations in which a principal (e.g., a company) wishes to assemble a team of agents to work on a project, but agents arrive online, and the principal must decide whether to hire each agent immediately upon arrival, without knowledge of future arrivals. This naturally gives rise to an online model, where uncertainty is captured by a known prior over agent availability, and information about which agents are actually available is revealed only as the process unfolds.

This perspective is closely related to the prophet inequality framework, where an online principal is compared to an omniscient ``prophet" that knows the entire realization in advance and can assemble the optimal team. The central question is therefore: how well can an online principal perform relative to this prophet?

Despite this conceptual connection, our setting differs fundamentally from standard prophet inequality models. Most notably, the principal is not constrained in the number of agents that may be hired, which would render classical prophet inequality problems trivial. Instead, hiring agents is costly, and the principal's objective is to maximize the reward generated by the team's effort minus the effort-inducing payments that depend on the team's composition. 
Moreover, the principal cannot directly control agents' actions. Rather, agents behave strategically, choosing effort levels that form an equilibrium and may exhibit free-riding. This strategic component raises new modeling questions absent from classical prophet inequalities, such as when agents choose their actions and what information they possess at that point.

Our contributions are twofold. First, we introduce a model for online contract design that combines sequential hiring decisions with strategic agent behavior. Second, we show that the online principal's ability to approximate the prophet depends jointly on the structure of the reward function and the timing of the agents' decisions.

\subsection{Our Model}
In our model, $n$ agents arrive sequentially in a fixed order. Each agent $i$ is active independently with probability $p_i$. An active agent may either exert effort or not, incurring a cost of $c_i$ if it exerts effort. The set $S$ of agents that exert effort determines the principal's expected reward through a set function $f$, where $f(S)$ denotes the expected reward generated by $S$. Upon the arrival of agent $i$, the principal offers a linear contract $\alpha_i \in [0,1]$, specifying that the agent receives an $\alpha_i$ fraction of the realized reward. The principal's utility is the total reward minus payments, given by
\[
\left(1-\sum_i \alpha_i\right)f(S),
\]
and the goal is to maximize the principal's expected utility.

Agents choose their behavior strategically, and we consider two timing models. In the first, agents choose their actions only after all contract offers have been made. The resulting actions form a Nash equilibrium, which determines the principal's utility. In the second, we study policies that induce dominant strategies. Under such policies, each agent decides whether to exert effort immediately after observing the contract offered upon its arrival, without waiting to see future contract offers or agent arrivals. We refer to such policies as \emph{safe policies}.

We evaluate an online policy by the expected utility it achieves for the principal and compare it to the expected utility of an offline (prophet) benchmark. The offline benchmark is the expected utility obtained by the optimal offline contract, which has complete knowledge of the active agents, as studied by~\cite{duetting2022multi}.

Interestingly, for unit-demand reward functions, our model admits a reduction to the classical single-choice prophet inequality problem, yielding a $2$-competitive policy. This guarantee is tight via a variant of the standard hard instance for prophet inequalities. Beyond the unit-demand setting, however, the two models diverge substantially. 
Already for $2$-demand (or additive) rewards, we we establish a $9/4$ lower bound (see~\Cref{sec:similarities}), highlighting the fundamentally different challenges posed by strategic online contract design.

\subsection{Our Results}

We study online contract design for the hierarchy of complement-free reward functions~\cite{lehmann2001combinatorial}. Our first result shows that for submodular rewards, there exists an online policy that achieves a constant-factor approximation to the prophet benchmark.

\medskip\noindent
\textbf{Main Result 1} (\Cref{thm:SM_const_CR})\textbf{.}
For submodular reward functions $f$, there exists an online policy that achieves an $O(1)$-approximation to the prophet benchmark.
\medskip

Our result is tight in two different senses. First, even for submodular rewards, no constant-factor approximation is possible if we restrict attention to safe policies. Second, for the more general class of XOS reward functions, no constant-factor approximation is possible even without restricting to safe policies.

\medskip\noindent
\textbf{Main Result 2} (\Cref{thm:submodular_gap})\textbf{.} For submodular rewards $f$ (or even coverage functions), no safe policy can guarantee better than $\Omega((\log n)/(\log \log(n))^2)$ approximation to the prophet benchmark.

\medskip\noindent
\textbf{Main Result 3} (\Cref{thm:xos})\textbf{.} For XOS rewards $f$, even without restricting to safe policies, no online policy can get better than $\Omega((\log \log n)/(\log \log \log n))$ approximation to the prophet.
\medskip

Notably, this separation between submodular and XOS rewards is unique to the online setting. In the offline contracting problem, the constant-factor approximation for submodular rewards extends to XOS rewards~\cite{duetting2022multi}. Likewise, in prophet inequalities for combinatorial auctions, the same (tight) $2$-approximation is known for both submodular and XOS valuations~\cite{FGL15,dutting2020prophet}.

A natural question is how far safe policies can be pushed. Specifically, what is the largest class of reward functions for which safe policies admit a constant-factor approximation? As a first step in this direction, we show that safe policies achieve a constant-factor approximation for weighted matroid rank functions (see~\Cref{sec:WMRF}).

\subsection{Our Techniques}
For our positive result for submodular $f$, we first observe that if a single agent approximates the offline optimum in expectation, then the problem reduces to the classic single-selection prophet inequality. Thus, the problem boils down to identifying which \emph{team} of ``small'' agents to hire (i.e., offer contracts $\alpha_i>0$), and how much to pay them.

To determine which agents to hire, we employ a pricing policy. Suppose a set $T^{(i-1)}$ has already been hired. When agent $i$ arrives, we hire it if its marginal contribution $\marg{i}{T^{(i-1)}} := f(T^{(i-1)}\cup \{i\}) - f(T^{(i-1)})$ exceeds a predetermined price $q_i$, while ensuring that the total reward of the hired set remains below a prescribed threshold (for a reason that will become apparent shortly). The prices follow the same general structure as those of~\cite{duetting2022multi} and are chosen to guarantee two key properties: (i) with constant probability, the hired team achieves a constant fraction of the optimal reward, and (ii) its total payment is at most one half.

Property (i) follows from viewing our pricing rule as an ``approximate online demand query'' with respect to $f$ and $\vec{q}$. We show that the hired set always achieves at least half the utility of the corresponding demand set, or a constant fraction of the prophet benchmark. 
In the former case, we connect the utility of the demand set to the prophet benchmark via the prices $\vec{q}$, using the same approach as in \cite{duetting2022multi}.
For this approach to succeed we need to condition on a good event, in which the optimal reward generated by the small agents is at least a constant fraction of its expected value, according to which $\vec{q}$ is defined.
To ensure that the good event occurs with constant probability we carefully define the set of small agents and apply concentration bounds for an appropriate XOS function defined over the set of small agents.

Property (ii) is achieved by bounding the total payments required to incentivize an approximate demand set. 
In the offline setting, previous work~\cite{duetting2022multi,feldman2025budget} utilized the fact that it is possible to dismiss some agents in the demand set. This is impossible in the online setting with irrevocable hiring decisions.
Instead, we bound the total payment as a function of the total reward generated by the hired agents and exploit the fact that our algorithm explicitly caps this reward throughout the hiring process.

Finally, we still need to decide \emph{how much} to pay each hired agent. We cannot generally ensure that all hired agents eventually exert effort in equilibrium, but we argue that it is possible to set prices so that the effort-exerting agents achieve comparable reward.
More concretely, assume agent $i$ is hired, and thus the (current) set of hired agents is $T^{(i)} = T^{(i-1)} \cup \{i\}$. We set the contract $\alpha_i$ to be \emph{twice} the optimal payment required to incentivize $i$ if $T^{(i)}$ was the set of effort-exerting agents.
This choice of $\alpha_i$ allows for some ``slack'', and exploits the fact that we  allow for some free-riding.
Indeed, let $W$ be the eventual set of effort-exerting agents. Any hired agent $i$ which does not exert effort has $\marg{i}{W} < \frac12 \marg{i}{T^{(i-1)}}$. This implies that the set of hired-but-shirking agents generates reward at most $f(T)/2$, where $T$ is the eventual set of hired agents.

For the negative result for submodular rewards the intuition is that the restriction to safe policies introduces a ``hard'' constraint, which is not present otherwise.
Namely, after the selection of any agent, we need to ensure that the marginal contribution does not drop to zero. 
This is because a necessary condition for an agent to exert in equilibrium is having strictly positive marginal contribution to the reward.
We exploit this by setting up an instance where the principal's reward is a coverage function on a laminar set system. 
Each agent corresponds to a set and having accepted an agent one cannot accept the agents corresponding to the supersets anymore. 
This leads to a tradeoff where the principal has to decide between accepting many small agents that arrive early, and few high-value agents that arrive later.
Note that, while a general policy can partially ``undo'' early decisions by tolerating some degree of free-riding, the laminar structure forces safe policies into tradeoffs that result in a super-constant loss.
The laminar set system we use is closely related to negative results on prophet inequalities and secretary problems for interval selection \cite{gobel2014online,im2011secretary,chawla2019pricing}.

For XOS, we draw inspiration from the construction that there is no constant-competitive prophet inequalities for the selection problem on general downward-closed set systems \cite{babaioff2007matroids,rubinstein2016beyond}.
There are $\sqrt{n}$ disjoint groups of $\sqrt{n}$ agents each. The function $f$ is defined to be the size of the largest number of selected agents belonging to one group.
As in the balls-into-bins problem, the largest number of active agents in a group is $\Theta((\log n) / (\log \log n))$.
We choose the costs such that the principal utility will definitely be negative when selecting agents from more than $\log n / \log \log n$ groups. As an online policy has to decide when seeing the first agent from a group, only the largest number of active agents in one of the $\log n / \log \log n$ selected groups counts, which is $\Theta((\log \log n) / (\log \log \log n))$

\subsection{Related Work}

\paragraph{Combinatorial Prophet Inequalities.}
The study of prophet inequalities originated with the classic work of \cite{krengel1977semiamarts,krengel1978semiamarts}, which introduced the problem of comparing an online stopping rule to a prophet with full foresight and established the optimal factor-2 guarantee. \cite{samuel1984comparison} showed that a simple threshold policy achieves the same worst-case guarantees. Subsequent work refined and generalized this framework; notably, \cite{kleinberg2012matroid} showed that the factor-2 bound extends to matroid constraints and connected prophet inequalities to online mechanism design, and \cite{rubinstein2017combinatorial} considered online combinatorial selection problems with submodular and XOS objectives. 
The exposition and simplification in \cite{dutting2020prophet} unified several results via threshold-based arguments and broadened their applicability. 
A related line of work develops online contention resolution schemes (OCRS) \cite{feldman2016online}, which provide a powerful abstraction for converting fractional offline solutions into online randomized algorithms with prophet-style guarantees.

\paragraph{Combinatorial Multi-Agent Contract Design.}
A combinatorial model for team-formation was introduced by \cite{babaioff2006combinatorial,BabaioffFNW12}, in which each of $n$ actions can work or shirk, and the principal's reward is given by a Boolean function. Follow up works \cite{babaioff2009free,babaioff2006mixed} considered free-riding and the power of mixed strategies in this model.
This setting was generalized by \cite{duetting2022multi}, where the principal's reward function $f$ is a combinatorial set function taken from the complement-free hierarchy of \cite{lehmann2001combinatorial}.
Their main result is an efficient algorithm that achieves a constant-factor approximation to the optimal contract when $f$ is XOS, using demand queries, and for submodular $f$ using value queries.
This result was shown to be tight by \cite{multimulti}, who proved that no PTAS exists for submodular rewards.
Building upon this multi-agent binary-action model, \cite{feldman2025budget,aharoni2025welfare} extend the above approximation guarantees for contracting under budget constraints and for a wide range of objective functions, beyond the principal's utility. 
\cite{gong2025approximating} extend the $O(1)$-approximation when the team of hired agents adheres to cardinality constraint.

The binary-action model was further generalized by \cite{multimulti}, where each agent can perform an arbitrary set from a personal collection of actions, and the reward is a combinatorial function of the actions taken in equilibrium.
They give poly-time $O(1)$-approximation to the optimal contract with demand queries.
\cite{dutting2025black} examine the assumption that agents engage in a Nash equilibrium and consider other solution concepts like correlated and coarse correlated equilibrium.
The works of \cite{feldman2026one} and \cite{feldman2026equal} study contract design under budget and fairness constraints, respectively.

\paragraph{Online Learning and Contract Design.}
An online perspective on contract design has also been studied through the lens of learning and regret minimization, beginning with \cite{ho2014adaptive}; see also \cite{cohen2023learning,ZhuBYWJJ23,BacchiocchiC0024,duetting2025pseudodimensioncontracts}.
This line of work studies how a principal learns good contracts when agents' costs, actions, or outcome distributions are initially unknown, typically in a single-agent contract environment.
Our setting is different in several aspects: First, the only unknown is which agents arrive, and the benchmark is the expected utility of the optimal offline contract.
Moreover, the nature of the problem is multi-agent, which introduces strategic considerations that are absent in the single-agent setting.

\paragraph{Bayesian Contract Design.}
There is a significant line of works that consider a principal which does not have full information about the agent's costs and actions-to-outcomes mapping, but only knows the distribution from which these traits (a.k.a, type) are drawn from \cite{alon2022bayesian, alon2021contracts,guruganesh2021contracts,GuruganeshSW023,CastiglioniM021,castiglioni2022designing,castiglioni2025reduction}.
This literature is mostly concerned with designing a (possibly randomized) mechanism that elicits the agent's type, while maximizing expected utility, where the expectation is taken over the type distribution as well.
This is fundamentally different from our setting, in which the agent's type is revealed upon arrival.

\paragraph{Concurrent Work.}
In a concurrent work \cite{lavi2026online} introduce an online non-Bayesian contracting model, where agents have binary effort levels, as in \cite{duetting2022multi}. 
In their model, the contracting game is comprised of two stages: In the selection stage, the principal decides at each point in time whether to include the arriving agent, and whether to irrevocably dismiss agents previously selected. Crucially, the principal does not know when the stream of agents will stop, and must maintain a ``good'' team of agents at any point.
In the second stage, the team is finalized and contracts are offered so that all selected agents are induced to exert effort in equilibrium. 
They design a tight $2$-competitive deterministic policy, with respect to the optimal principal's utility in hindsight, where the reward function $f$ is additive. When $f$ is XOS they show that even a randomized policy cannot achieve a constant competitive ratio.

\section{Model and Preliminaries}
In this section we present the online contract-design model studied in the paper, followed by the definition of \emph{\safe} policies, the classes of reward functions we consider, and a short toolbox of useful lemmas. The offline optimal that we compare against is the same objective studied in \cite{duetting2022multi}.

\subsection{Online model and timeline}

We consider the binary-outcome multi-agent setting initiated by \cite{duetting2022multi}. A principal wishes to delegate a task to a subset of agents from the ground set $[n]=\{1,\dots,n\}$. The task can either succeed or fail. Each agent $i$ can exert effort or shirk. Exerting effort comes at a cost of $c_i \geq 0$, while shirking incurs no cost.  The task's success probability is a combinatorial function $f:2^{[n]}\to [0,1]$, which maps sets of agents $S$ that exert effort to a success probability $f(S)$. We generally assume that $f$ is normalized so that $f(\emptyset) = 0$ and monotone so that $S \subseteq T$ implies $f(S) \leq f(T)$. For brevity we denote the singleton value by $f_i=f(\{i\})$.

If the task succeeds the principal receives a reward of $r \geq 0$, otherwise the principal's reward is zero. We normalize the principal's reward upon success to $1$. Under this assumption, $f$ can also be interpreted as the principal's expected reward function.

This basic setup, where the principal enjoys the reward from effort but the agents bear the cost, creates a situation, which economists refer to as \emph{moral hazard}. To mitigate this misalignment, the principal offers each agent $i$ a \emph{linear contract} $\alpha_i \in [0,1]$, which specifies that upon success a fraction $\alpha_i$ of the realized reward should be transferred to agent $i$.

We consider an online Bayesian/prophet-style model where agents arrive one-by-one in fixed adversarial order and each agent $i$ is \emph{active} with probability $p_i\in[0,1]$, independently. We index agents so that agent $i$ is the $i$-th agent to arrive. 
The principal knows the arrival order, the function $f$, the costs $\{c_i\}_{i\in[n]}$, and the activation probabilities $\{p_i\}_{i\in[n]}$ in advance, but does not 
observe which agents are active until they arrive.
Upon arrival of agent $i$ the principal must decide, immediately and irrevocably, which contract $\alpha_i$ to offer.
We refer to the algorithm that the principal uses to determine these contracts, which may depend on the history and may be randomized, as the \emph{online policy} $\pi$ of the principal. We denote by $A$ the set of active agents, and use $\contract = (\alpha_i)_{i \in A}$ to denote the contracts offered to the active agents.

At the end of the process, after all agents have arrived and all offers have been made, the active agents play a Nash equilibrium of the game induced by $\contract$, in order to decide which exerts effort. As standard, we assume that agents break ties in favor of the principal. 
Let $S \subseteq A$ denote the set of agents that exert effort. 
The expected utility of agent $i$ is $u_i(\contract,S) = \alpha_if(S) - c_i$ if $i \in S$ and it is $u_i(\contract,S) = \alpha_i f(S)$ if $i \not\in S$. We say that $S$ is an equilibrium if $u_i(\contract,S) \geq u_i(\contract,S\setminus \{i\})$ for $i \in S$ and $u_i(\contract,S) \ge u_i(\contract,S \cup \{i\})$ for $i \not\in S$.
Equivalently, $\alpha_i \ge \frac{c_i}{f_i(S)}$ for any $i \in S$ and $\alpha_i \le \frac{c_i}{f_i(S)}$ for any $i \in A \setminus S$, where $\marg{i}{S} = f(S \cup\{i\}) - f(S\setminus\{i\})$ is the marginal contribution of agent $i$ to the reward when working together with $S\setminus\{i\}$. 
We note that for every contract $\contract$ there is an equilibrium (this is because the induced game is a potential game \cite{deo2024supermodular}).

As usual, we assume tie-breaking in favor of the principal, and we denote by $S(\contract) \subseteq A$ the equilibrium that maximizes the principal's utility, $(1-\sum_{i \in A} \alpha_i)f(S)$.
We define the expected principal's utility induced by policy $\pi$ as
\[
ALG(\pi)=\mathbb{E}_{A,\pi}\bigg[\bigg(1-\sum_{i\in A}\alpha_i\bigg)\cdot f\big(S(\contract)\big)\bigg],
\]
We compare $ALG(\pi)$ to the offline optimum
\[
OPT=\mathbb{E}_{A}\Big[\max_{S\subseteq A} g(S)\Big],\qquad g(S):=\Big(1-\sum_{i\in S}\frac{c_i}{\marg{i}{S}}\Big)\cdot f(S),
\]
which is exactly the objective in \cite{duetting2022multi}: given full knowledge of which agents are active, the principal can choose a contract to incentivize a working set $S\subseteq A$ by setting $\alpha_i=c_i/\marg{i}{S}$ for $i\in S$, and $\alpha_i=0$ for $i\notin S$. 
The optimal offline contract thus reduces to maximizing $g(S)$ over subsets of $A$ (interpreting $c_i/\marg{i}{S}=0$ when $c_i=\marg{i}{S}=0$ and $=\infty$ when $\marg{i}{S}=0<c_i$).

We say an online policy $\pi$ achieves competitive ratio $\rho \ge 1$ if $OPT/ALG(\pi)\le \rho$ in the worst case over instances $\langle f,\{c_i\},\{p_i\}\rangle$.

\begin{remark}
From this point onwards, we focus on the case where $c_i > 0$ for all agents $i \in [n]$. 
This is the interesting case because for submodular $f$ (see definition of submodularity below), 
if we can get a $\rho$-competitive policy under this assumption, we can get a $(2\rho)$-competitive algorithm without this assumption (see Proposition~\ref{prop:positive_costs} in Appendix~\ref{sec:positive_costs}). Beyond submodular $f$ we mostly show negative results, which naturally extend to settings where costs can be zero. 
A main advantage of focusing on settings with positive costs is that it isolates an interesting effect, which is cleanest in this regime. Namely, it allows to distinguish between agents that are ``hired'' (get a positive contract $\alpha_i > 0$) and those that ``actually exert effort'' (those agents that belong to $S(\contract)$). The positive cost assumption ensures that only agents that are hired may have an incentive to actually exert effort. Without this assumption we may have agents that receive a zero contract ($\alpha_i = 0$), and in that sense are not hired, but still decide to work.
\end{remark}

\subsection{\Safe{} policies}

We refer to agents that receive a contract $\alpha_i > 0$ as \emph{hired}, and to agents $i \in S(\contract)$ as \emph{working}.
Roughly speaking, a \emph{\safe} policy never \emph{hires} an agent unless that agent will in equilibrium \emph{work} (i.e., exert effort). Equivalently, \safe{} policies avoid free-riders: every agent offered a non-zero contract indeed exerts effort in the resulting equilibrium.

 We emphasize that restricting attention to \safe{} policies can be a loss of generality (as we show in \Cref{obs:not-safe} and \Cref{thm:submodular_gap}), but \safe{} policies are natural, and  also facilitate stronger solution concepts when $f$ is submodular.

\begin{definition}[\Safe{} policy]
	An online policy $\pi$ is \emph{\safe} if for every realization of active agents $A$ and for every contract vector $\alpha$ produced by $\pi$, the set of hired agents
	$S^{\mathrm{hire}}:=\{i\in A\mid \alpha_i>0\}$ satisfies
	\[
	\alpha_i \ge \frac{c_i}{\marg{i}{S^{\mathrm{hire}}}}\qquad\text{for every }i\in S^{\mathrm{hire}}.
	\]
	In words: each hired agent receives at least the payment required to make effort a (weakly) best-response against the other hired agents.
\end{definition}

Under the positive costs assumption and submodular $f$, \safe{} policies induce (weakly) dominant strategies for each agent: if an agent is hired she will prefer exerting effort regardless of other agents' choices.
This property fails to hold when $f$ is XOS, as we show in \Cref{prop:safe_XOS_not_DS}. The proof of the observation below and the rest of the details are deferred to \Cref{sec:safe_equilibrium}.

\begin{restatable}{observation}{safeimpliesds}\label{obs:safe_implies_ds}
Assume that $f$ is monotone submodular and that $c_i>0$ for every agent $i$. Then every safe policy induces exerting effort as a weakly dominant strategy for every hired agent.
\end{restatable}

We also observe that in general, restricting attention to safe policies entails a loss of generality. Specifically, there exist instances in which the optimal policy is not safe, and no safe policy achieves the same performance.

\begin{observation}
    [The optimal policy is generally not safe]
    \label{obs:not-safe}
    Let $\eps \in (0, 0.2)$ and consider the following instance with two agents and a unit-demand $f$, i.e., $f(S)=\max_{i\in S}f(\{i\})$.
    Agent $1$ is always active, i.e., $p_1=1$, with $c_1 = \eps^2$ and $f(\{1\})= \eps$.
    Agent $2$ is active with probability $p_2=\eps$, with $c_2 = \eps$ and $f(\{2\}) = 1$.
    Consider any safe policy $\pi$. It will not set $\alpha_1 > 0$ and $\alpha_2 > 0$ simultaneously. If $\alpha_1 > 0$, then the optimal contract is $\alpha_1 = \eps$, and so the principal utility will be $(1 - \eps) \cdot \eps$. Otherwise, if $\alpha_1 = 0$, the highest possible expected principal utility will be $\eps \cdot (1 - \eps)$, namely by setting $\alpha_2 = c_2/f_2$ whenever agent $2$ is active. 
    Compare this to a policy that always sets $\alpha_1 = \eps$ and that sets $\alpha_2 = 2 \eps/(1 - \eps)$ if agent $2$ is active, and $\alpha_2 = 0$ otherwise. 
    The final equilibrium is unique: If  agent $2$ is active, it will exert effort. Otherwise, agent $1$ will exert effort. 
    The competitive ratio between the two policies is
    \[
    \frac{\eps \cdot (1 - \eps - 2 \eps/(1 - \eps)) \cdot 1 + (1 - \eps) \cdot (1 - \eps) \cdot \eps}{\eps(1-\eps)} \ge \frac{\eps(2-6\eps)}{\eps(1-\eps)} \longrightarrow 2
    \]
\end{observation}

\subsection{Classes of reward functions}
We now list the valuation classes used in the paper, in increasing order of generality.
\begin{description}
	\item[Additive:] $f$ is additive if there exist non-negative weights $f_i$ with $f(S)=\sum_{i\in S} f_i$ for all $S$. 
	\item[Coverage:] $f$ is coverage if there exists a finite universe $U$, with weights $w_u \ge 0$ for each $u \in U$, such that each agent $i$ covers a subset $U_i\subseteq U$ and $f(S)$ is the weighted sum of covered elements. I.e.,  $f(S)=\sum_{u \in U(S)} w_u$, where $U(S) = \bigcup_{i\in S} U_i$. Coverage functions are submodular.
	\item[Submodular:] $f$ is submodular if for every $S\subseteq T$ and $i\notin T$ we have $\marg{i}{S}\ge \marg{i}{T}$. Submodular functions capture diminishing returns.
	\item[XOS (fractionally subadditive):] $f$ is XOS if it can be written as the pointwise maximum of additive functions, i.e., $f(S)=\max_{\ell\in L}\sum_{i\in S} w_{\ell,i}$ for some collection of additive clauses indexed by $L$. 
	\item[Subadditive:] $f$ is subadditive if for every $S,T$ we have $f(S\cup T) \le f(S)+f(T)$. This is the most general class we consider.
\end{description}

When $f$ is submodular, the induced offline objective $g(S)=\big(1-\sum_{i\in S} c_i/\marg{i}{S}\big)\cdot f(S)$ is subadditive; see Lemma~\ref{lem:gSA}.

\subsection{Toolbox}
We build on several useful observations about the offline contracting problem which were made in prior work.
We use the following well-known property of XOS functions to apply concentration bounds to the reward of a random set of agents.
We prove an expected-value version of this lemma and defer the proof to \Cref{sec:expected-xos-margs-proof}.
\begin{lemma}[XOS Marginals Lemma \cite{duetting2022multi,fu2012conditional}]\label{lem:XOS_margs}
	For any XOS function $f$ and any sets $S\subseteq T$,
	$\sum_{i\in S} \marg{i}{T} \le f(S)$.
\end{lemma}

\begin{lemma}[XOS Expected Marginals Lemma]\label{lem:expected_XOS_margs}
	Let $f:2^{[n]}\to\reals_{\ge 0}$ be an XOS function. Let $A\subseteq[n]$ be a random set in which each agent $i$ is included independently with probability $p_i$. Then
	\[
		\E_A[f(A)]
		\ge
		\sum_{i=1}^n p_i\cdot
		\E_{A_{-i}}\big[f(A_{-i}\cup\{i\})-f(A_{-i})\big],
	\]
	where $A_{-i}=A\setminus\{i\}$ is a random set obtained by independently sampling all agents but $i$.
\end{lemma}

The following lemma connects the sum of squared costs to the reward of a set. This connection is used, as in \cite{duetting2022multi}, to design prices such that the reward of a demand set provides a constant fraction of the optimal reward. It is a budgeted refinement of Lemma 3.3 in \cite{duetting2022multi}, which is recovered by taking $\rho=1$.
\begin{lemma}[Budgeted refinement of Lemma 3.3 in \cite{duetting2022multi}]\label{lem:budgeted_sqrt_cost}
	Let $f$ be XOS and let $S\subseteq[n]$ satisfy $\marg{i}{S}>0$ whenever $c_i>0$.
	If $\sum_{i\in S}\frac{c_i}{\marg{i}{S}} \le \rho$ for some $\rho>0$, then
	$\sum_{i \in S}\sqrt{c_i} \le \sqrt{\rho \cdot f(S)}$.
\end{lemma}
\begin{proof}
	Agents with $c_i=0$ contribute nothing to the left-hand side, so assume $c_i>0$, hence $\marg{i}{S}>0$, for all $i \in S$. By Cauchy--Schwarz,
	\[
	\sum_{i\in S}\sqrt{c_i}
	=
	\sum_{i\in S}\sqrt{\frac{c_i}{\marg{i}{S}}}\cdot\sqrt{\marg{i}{S}}
	\le
	\bigg(\sum_{i\in S}\frac{c_i}{\marg{i}{S}}\bigg)^{1/2}
	\bigg(\sum_{i\in S}\marg{i}{S}\bigg)^{1/2}
	\le
	\sqrt{\rho}\cdot\bigg(\sum_{i\in S}\marg{i}{S}\bigg)^{1/2}.
	\]
	By \Cref{lem:XOS_margs} (applied with $S=T$), $\sum_{i\in S}\marg{i}{S} \le f(S)$, and the claim follows.
\end{proof}

The following exemplifies that when $f$ is submodular, the principal's utility function $g$ is subadditive. The proof is deferred to \Cref{sec:gSA}.
\begin{restatable}{lemma}{gsubadditive}\label{lem:gSA}
	Whenever $f$ is monotone and submodular, the function $g$ is subadditive.
\end{restatable}

We also use the following one-sided Chebyshev's inequality.
\begin{lemma}[Cantelli's inequality]\label{fact:cantelli}
	Let $Z$ be a random variable with mean $m$ and variance $\sigma^2$. For every $t>0$,
	$\Pr[Z \le m - t] \le \frac{\sigma^2}{\sigma^2+t^2}$.
\end{lemma}

\section{Similarities and Differences to the Classic Prophet Inequalities}\label{sec:similarities}

We compare our model to the classic prophet inequality and present two fundamental results.
First, when $f$ is unit-demand, i.e., $f(S) = \max_{i \in S}f_i$, the principal's problem reduces to the standard single-item prophet inequality with $n$ items.
In particular, a (tight) competitive ratio of $2$ can be achieved.
\begin{proposition}[see \Cref{sec:unit-demand}]
    When $f$ is unit-demand, there exists a $2$-competitive online policy.
    Moreover, for every $\eps > 0$, there exists a unit-demand instance for which any online policy achieves a competitive ratio of at least $2-\eps$.
\end{proposition}

Observe that incentivizing a single agent $i^\star$ is done optimally using $\alpha_{i^\star} = \frac{c_{i^\star}}{f_{i^\star}}$. This yields utility $(1-\frac{c_{i^\star}}{f_{i^\star}})f_{i^\star} = f_{i^\star}-c_{i^\star}$. 
Thus, a valid policy is to pick one of $n$ agents (boxes), each available with probability $p_i$, each with a value of $v_i=\max\{(f_i-c_i),0\}$. As such, a competitive ratio of $2$ can be achieved by a simple threshold policy \cite{samuel1984comparison,kleinberg2012matroid}.
The matching lower bound in our setting does not carry over immediately, as the principal may deploy a policy which is not safe, and hire both agents. However, minor adaptations to the well-known hard instance provide a matching 2-competitive lower bound. We defer the details to \Cref{sec:unit-demand}.

We also observe that the tight connection between the two problems disappears even for $2$-demand $f$, i.e., $f(S) = \max_{S' \subseteq S, |S'|\le 2} \sum_{i \in S'} f_i$.
In this case, we present an instance in which any policy achieves a competitive ratio at least $9/4$, which implies a strict separation from the online combinatorial allocation problem, in which a competitive ratio of $2$ can be achieved even when valuations are XOS \cite{FGL15}.

\begin{restatable}[see \Cref{sec:add_lower_bound}]{proposition}{additivelowerbound}
    For $2$-demand or additive $f$, there exists an instance of the online contracting problem 
    for which any algorithm achieves a competitive ratio of at least $9/4$.
\end{restatable}

In the standard argument for the lower bound of $2$ in prophet inequalities, the gambler has to choose a-priori between receiving a reward of $1$ with probability $1$ (taking the first agent/box) and a reward of $1/\eps$ with probability $\eps$ (the second agent/box).
While the prophet, which observes the realization can target the high reward whenever it is available.

We augment this construction by adding two more agents \emph{at the beginning} of the process, where the first is deterministic (i.e., $p_1=1$), and the other is stochastic.
We carefully choose the activation probabilities, rewards, and costs so  that the principal must pick, a-priori, one of  three mutually exclusive alternatives, all with the same expected value: 
(i) incentivize agent $1$ to exert effort, hoping that agent $2$ will be active as well, 
(ii) skip $1$ and $2$, and incentivize agent $3$, and
(iii) skip $1,2$ and $3$, and incentivize $4$, if active.
The prophet, which can observe whether agents $2$ and $4$ are active or not, can guarantee the maximum of the three.
We bring the full details in \Cref{sec:add_lower_bound}.

\newcommand{\OPT}{S^\star}
\newcommand{\cheapAgents}{L}
\newcommand{\OPTbudget}{L^\star}
\newcommand{\smallAgents}{X}
\newcommand{\budgetHalfDefinition}{\OPTbudget \in \argmax_{S \subseteq A} f(S) \;\text{ s.t. }\; \sum_{i \in S} \frac{c_i}{\marg{i}{S}} \le \frac12}

\section{\boldmath Constant Competitive Ratio for Submodular $f$}\label{sec:positive_SM}

In this section we design an online policy which achieves a constant
competitive ratio, compared to the offline optimum, when the reward function
$f$ is submodular.  Notably, our algorithm achieves the same guarantees even
when the arrival order is unknown a-priori and chosen adversarially.

\begin{theorem}\label{thm:SM_const_CR}
When $f$ is submodular, there exists an online policy which achieves a
competitive ratio of $174$.
\end{theorem}

To this end, we reduce the problem into two sub-problems, while maintaining a
constant competitive ratio: (i) Designing an $O(1)$-competitive policy for the
best single-agent contract. (ii) Designing an $O(1)$-competitive policy to an
online budgeted contracting problem (see below).
While problem (i) can be solved using standard prophet inequality algorithms for the canonical single-item setting \cite{krengel1977semiamarts,samuel1984comparison,kleinberg2012matroid}, as we show in \Cref{sec:unit-demand}, problem (ii) is more challenging and requires new ideas.
In \Cref{subsec:budgeted_online,sec:good_event} we present an online policy that achieves a constant competitive ratio for the budgeted problem, and conclude with the proof of \Cref{thm:SM_const_CR}.

\paragraph{A budgeted variant.}
We define an online budgeted contracting problem, which is identical to the
online contracting problem, except that the principal is trying to maximize the
value of $f(S)$ (where $S$ is the best pure Nash equilibrium induced by the
payments), subject to total payments not exceeding $1/2$.  For a fixed set of
active agents $A$, we denote the offline optimum by $L^\star$.
\begin{equation}\label{eq:budget_half}
\budgetHalfDefinition
\end{equation}
Observe that for submodular $f$ it is without loss of generality to only
consider agents from $L = \{i \in [n] \mid c_i/f(\{i\}) \le 1/2\}$.  This is
because agents outside $L$ can only be incentivized using a contract exceeding
$1/2$, which would violate the budget constraint.

The reduction from the online contracting problem to the online budgeted contracting problem is formalized in the following proposition.
We defer the proof of \Cref{prop:reduction_to_budget} to \Cref{subsec:decomposition}.
\begin{restatable}{proposition}{reductiontobudget}\label{prop:reduction_to_budget}
    Any $\gamma$-competitive policy for the online budgeted contracting problem, implies a $3(\gamma+2)$-competitive policy for the online contracting problem.
\end{restatable}

\subsection{The Online Budgeted Problem}\label{subsec:budgeted_online}

Let $\mu^\star = \E[f(L^\star)]$ be the expected offline optimum for the
budgeted problem (\Cref{eq:budget_half}).  In order to design a competitive
policy with respect to this benchmark, we look at the same problem but with a
restricted set of agents: those whose individual reward is no more than $\mu^\star / 14$.
More formally, for a realized set of agents $A$, let
$X = \{ i \in A \mid c_i/f_i \le 1/2 \text{ and } f_i \le \mu^\star /14\}$.  We
denote by $X^\star$ the optimal set for the offline problem, while restricting
attention to agents in $X$.
\begin{equation}\label{eq:_Xstar}
    X^\star \in \argmax_{S \subseteq X} f(S) \text{ s.t. } \sum_{i \in S} \frac{c_i}{\marg{i}{S}} \le \frac{1}{2}.
\end{equation}

Let $\nu^\star = \E[f(X^\star)]$.  We use one of two algorithms, according to
the ratio $\nu^\star / \mu^\star$ (observe that it is upper bounded by $1$).

If $\nu^\star / \mu^\star < 13/14$, we show in \Cref{subsec:budget_single_agent} that it is enough to pick at most one agent from $L \setminus X$.  We use standard prophet inequality techniques,
formulating an ex-ante relaxation of this problem as an LP, and approximating this stronger benchmark with an OCRS.

If $\nu^\star / \mu^\star \ge 13/14$, a single agent may not be sufficient to
well approximate $\mu^\star$, and in \Cref{subsec:budget_multi_agent} we
present an algorithm for online team contracting, under budget constraints.

\subsubsection{Incentivizing a Single Agent}\label{subsec:budget_single_agent}

In this section we consider the problem of online hiring a single agent from
the following set
$$
\max_{i \in L \setminus X} g(\{i\}) \qquad \text{where }
L \setminus X = \{i \in A \mid c_i/f_i \le 1/2 \text{ and } f_i > \mu^\star / 14\}$$
We consider the ex-ante relaxation of the offline optimum.  It is well-known
that there is an OCRS which guarantees $1/2$ of the optimal solution to the LP
below.
\begin{align}\label{eq:large_agent_LP}
\max \quad & \sum_{i \in L \setminus X} z_i \, g(\{i\})
\qquad \text{s.t.}
\qquad \sum_{i \in L \setminus X} z_i \le 1
\quad \text{and} \quad z_i \ge 0 \quad \forall i \in L \setminus X
\end{align}

We show that if $\nu^\star < \frac{13}{14}\mu^\star$, hiring using the algorithm
of \cite{alaei2014bayesian} achieves an $O(1)$-CR for the budgeted problem.

\begin{proposition}\label{prop:muStar_largeAgentCase}
If $\nu^\star < \frac{13}{14}\mu^\star$, then running the OCRS of
\cite{alaei2014bayesian} w.r.t.~the optimal solution of
\Cref{eq:large_agent_LP} induces a single agent whose payment is at most one-half, and whose expected reward is at
least $\mu^\star / 56$.
\end{proposition}

The correctness of \Cref{prop:muStar_largeAgentCase} follows immediately from
the $2$-competitiveness guarantee of \cite{alaei2014bayesian}, and the lemma
below.

\begin{lemma}\label{lem:largeAgent_LP_value}
Assume $\nu^\star < \frac{13}{14}\mu^\star$.  If $z^\star_i$ is the optimal
solution to (\ref{eq:large_agent_LP}), then
$\sum_{i \in L\setminus X}z^\star_i \cdot g(\{i\}) \ge \mu^\star / 28$.
\end{lemma}
\begin{proof}
Recall that $p_i$ is the probability that agent $i$ is active.  Let
$\vec{p} \in [0,1]^{|L\setminus X|}$ denote the vector of probabilities of the
agents in $L\setminus X$.  We consider two cases.

If $\sum_{i \in L\setminus X} p_i \le 1$, then $\vec{p}$ is a feasible solution
to (\ref{eq:large_agent_LP}).
\begin{align*}
    \sum_{i \in L\setminus X}z^\star_i \cdot g(\{i\})
    &\ge
    \sum_{i \in L\setminus X}p_i \cdot g(\{i\}) \\
    &\ge
    \sum_{i \in L\setminus X}\Pr[i \in L^\star] \cdot g(\{i\}) && (\text{whenever $i \in L^\star$, agent $i$ is active}) \\
    &\ge
    \E_{A}[g(L^\star \setminus X)] && (\text{subadditivity of $g$}) \\
    &\ge
    \frac12 \E_{A}[f(L^\star \setminus X)] && (L^\star \setminus X\text{ is a feasible solution to (\ref{eq:_Xstar})}) \\
    &\ge
    \frac12(\E_{A}[f(L^\star)] - \E_{A}[f(L^\star \cap X)]) && (\text{subadditivity of $f$}) \\
    &\ge
    \frac12(\mu^\star - \nu^\star) && (\text{optimality of $X^\star$}) \\
    &>
    \frac12\cdot\frac{\mu^\star}{14}
    =
    \frac{\mu^\star}{28} && \bigg(\nu^\star < \frac{13}{14} \mu^\star \bigg)
\end{align*}
Otherwise, let $P = \sum_{i \in L\setminus X} p_i$.  By assumption $P>1$, and
the vector $\vec{p'} = \frac{\vec{p}}{P}$ is a feasible solution to
(\ref{eq:large_agent_LP}), as $\sum_i p'_i = \sum_i \frac{p_i}{P} = 1$.
\begin{align*}
    \sum_{i \in L\setminus X}z^\star_i \cdot g(\{i\})
    \ge
    \sum_{i \in L\setminus X}p'_i \cdot g(\{i\})
    \overset{(1)}{\ge}
    \sum_{i \in L\setminus X}p'_i \cdot \frac12 \cdot f(\{i\})
    \overset{(2)}{>}
    \sum_{i \in L\setminus X}p'_i \cdot \frac{\mu^\star}{28}
    =
    \frac{\mu^\star}{28}
\end{align*}
where $(1)$ and $(2)$ use the fact that for any $i \in L \setminus X$,
$\frac{c_i}{f_i} \le \frac12$ and $f_i > \frac{\mu^\star}{14}$, respectively.
\end{proof}

\subsubsection{Incentivizing a Team of Agents}\label{subsec:budget_multi_agent}

We now address the main technical challenge of this section.  Namely, design an
online policy which only hires agents from
$X = \{i \in A \mid c_i/f_i \le 1/2,\ f_i \le \mu^\star / 14\}$, while adhering
to the budget constraint of $1/2$, and whose expected utility is at least
$c \cdot \mu^\star$, for some constant $c>0$.

Our guarantees in this section hold under two assumptions.  First,
$\nu^\star \ge \frac{13}{14}\mu^\star$ (with
\Cref{subsec:budget_single_agent} handling the complementary case).  Second, the
realized reward of the best set $X^\star \subseteq X$ (see \Cref{eq:_Xstar}) is
at least $f(X^\star) \ge \tfrac58\,\nu^\star$.
We term this event the ``good event''.  In \Cref{sec:good_event} we establish,
by applying the Efron--Stein inequality \cite{efron1981jackknife}, that the good
event occurs with probability at least $117/181$.

\begin{proposition}\label{prop:muStar_smallAgentCase}
If $\nu^\star \ge \frac{13}{14}\mu^\star$, then \Cref{alg:small_agents_SM}
achieves a competitive ratio of $55$ for the online budgeted problem.  That is,
the payments made in any realization are at most $1/2$, and it yields an
expected principal's utility of at least $\mu^\star / 55$.
\end{proposition}

\Cref{alg:small_agents_SM} operates as follows.  It maintains a set of hired
agents, which are offered a non-zero contract; we denote this set by $T^i$ after
the arrival of agent $i$.  An active agent $i \in X$ is offered a non-zero
contract if its marginal contribution with respect to $T^{i-1}$ is large enough,
namely $\marg{i}{T^{i-1}} \ge \frac{2}{\sqrt5}\sqrt{c_i \nu^\star}$, and the
total reward is not too large, $f(T^{i-1}\cup\{i\}) \le \nu^\star /5$.  In this
case, the non-zero contract offered to $i$ is
$\alpha_i = \frac{2c_i}{\marg{i}{T^{i-1}}}$.  Note that this is twice the
optimal contract for $i$ if the set of effort-exerting agents is
$T^{i-1}\cup \{i\}$.  In any other case, agent $i$ is not hired, i.e.,
$\alpha_i=0$.  We denote the final set of agents which receive a non-zero
contract by $T$.

\begin{algorithm}[t]
\caption{Online Selection Algorithm for Small Agents}
\label{alg:small_agents_SM}

\KwIn{$n$, costs $\{c_i\}_{i \in [n]}$, the values $\mu^\star$ and $\nu^\star$, the set $X  = \{i \in N \mid \frac{c_i}{f_i}\le \frac12 \land f_i \le \frac{\mu^\star}{14}\}$.}
\KwOut{A vector of contracts $\contract \in [0,1]^{A}$, and $T = \{i \in A \mid \alpha_i >0\}$}

$T^0 \gets \emptyset, \contract \gets \textbf{0}$\;

\For{$i = 1$ \KwTo $n$}{
    $q_i \gets \frac{2}{\sqrt5}\sqrt{c_i \cdot \nu^\star}$\;

    \eIf{$i$ is active \textbf{and} $i \in \smallAgents$ \textbf{and} $f_{T^{(i-1)}}(i) \ge q_i$ \textbf{and} $f(T^{i-1} \cup \{i\}) \le \nu^\star/5$}{
        $\alpha_i \gets \frac{2c_i}{\marg{i}{T^{i-1}}}$,
        $T^{i} \gets T^{i-1} \cup \{i\}$ \;
    }{
        $\alpha_i \gets 0$,
        $T^{i} \gets T^{i-1}$\;
    }
}
$T \gets T^n$\;
\Return $\contract$, $T$\;

\end{algorithm}

We now turn to analyze \Cref{alg:small_agents_SM}, and begin by showing that it
respects the budget constraint in any realization.

\begin{lemma}\label{lem:smallAgents_bounded_payment}
For any realized set of active agents $A$, the contract $\contract$ produced by
\Cref{alg:small_agents_SM} satisfies $\sum_{i \in A} \alpha_i \le 1/2$.
\end{lemma}
\begin{proof}
Recall that $T \subseteq A$ is the set of agents for which $\alpha_i>0$.  For
any $i \in T$, the marginal with respect to the intermediate set
$T^{i-1} \subseteq T$ is at least the price
$q_i= \frac{2}{\sqrt5} \sqrt{c_i \cdot \nu^\star}$.  Thus,
\begin{align*}
    \sum_{i \in A} \alpha_i
    =
    \sum_{i \in T} \alpha_i
    =
    \sum_{i \in T} \frac{2c_i}{\marg{i}{T^{i-1}}}
    \le
    \sum_{i \in T} \frac{2c_i}{q_i}
    =
    \sum_{i \in T} \frac{2c_i \sqrt5}{2\sqrt{c_i\nu^\star}}
    =
    \frac{\sqrt5}{\sqrt{\nu^\star}} \cdot \sum_{i \in T} \sqrt{c_i}.
\end{align*}
To bound the right hand side we again use
$\marg{i}{T^{(i-1)}} \ge q_i = \frac{2}{\sqrt5} \sqrt{c_i \cdot \nu^\star}$,
which rearranges to
$\sqrt{c_i} \le \frac{\sqrt5}{2}\cdot\frac{\marg{i}{T^{(i-1)}}}{\sqrt{\nu^\star}}$.
Hence
\begin{align*}
    \sum_{i \in T} \sqrt{c_i}
    \le
    \frac{\sqrt5}{2\sqrt{\nu^\star}}\sum_{i \in T} \marg{i}{T^{(i-1)}}
    = \frac{\sqrt5}{2\sqrt{\nu^\star}}
    \sum_{i \in T} \big(f(T^{(i)}) - f(T^{(i-1)})\big)
    = \frac{\sqrt5\, f(T)}{2\sqrt{\nu^\star}}
    \le \frac{\sqrt5}{2\sqrt{\nu^\star}}\cdot\frac{\nu^\star}{5}
    = \frac{\sqrt{\nu^\star}}{2\sqrt5},
\end{align*}
where the last inequality follows since $f(T) \le \nu^\star / 5$.  Combining the
two inequalities above gives,
\[
\sum_{i \in A} \alpha_i
\le
\frac{\sqrt5}{\sqrt{\nu^\star}}\cdot\frac{\sqrt{\nu^\star}}{2\sqrt5}
=
\frac12,
\]
which concludes the proof.
\end{proof}

In the following lemma we show that even though some agents in $T$ may not exert
effort under the best equilibrium, the set of agents that do exert effort yields reward of
at least $f(T)/2$.  

\begin{lemma}\label{lem:smallAgents_lowerbound_fW}
Let $W \subseteq T$ be the set of agents that work in the best equilibrium,
i.e., $i \in W$ if and only if $\alpha_i \ge \frac{c_i}{\marg{i}{W}}$.  Then
$f(W) \ge \frac{1}{2}f(T)$.
\end{lemma}
\begin{proof}
Let $T\setminus W = \{s_1,\dots,s_l\}$, where the order is aligned with the
order of appearance in \Cref{alg:small_agents_SM}.  Observe that for any
$s \in T \setminus W$,
$\frac{2c_s}{\marg{s}{T^{s}}} = \alpha_{s} < \frac{c_s}{\marg{s}{W}}$, thus
$\marg{s}{W} < \frac12 \marg{s}{T^{s}}$.
\begin{align*}
    f(T)
    &=
    f(W) + \sum_{j=1}^l f(W \cup \{s_1,\dots,s_{j}\})-f(W \cup \{s_1,\dots,s_{j-1}\}) && (\text{telescopic sum}) \\
    &\le
    f(W) + \sum_{j=1}^l \marg{s_j}{W} && (\text{submodularity}) \\
    &<
    f(W) + \frac12 \sum_{j=1}^l \marg{s_j}{T^{s_j}} \\
    &\le
    f(W) + \frac12 \sum_{i \in T} \marg{i}{T^{i}} && (\{s_1,\dots,s_l\} \subseteq T) \\
    &=
    f(W) + \frac12 f(T) && (\text{telescopic sum}).
\end{align*}
This completes the proof.
\end{proof}

Next, we establish that under the good event, 
i.e., whenever the set $T$ yields high reward.

\begin{lemma}\label{lem:smallAgents_lowerbound_fT}
Fix a realized set of active agents $A$. If both $f(X^\star)\ge \frac58\nu^\star$
and $\nu^\star \ge \frac{13}{14}\mu^\star$ hold, then
$f(T) \ge \frac{4}{35}\mu^\star$.
\end{lemma}
\begin{proof}
Let $A$ be the realized set of active agents.  Assume
$f(X^\star)\ge \frac58\nu^\star$ and $\nu^\star \ge \frac{13}{14}\mu^\star$.  We
consider two cases.

\medskip\noindent
\textbf{Case 1:} For every agent $i \in \smallAgents \cap A$, if
$\marg{i}{T^{(i-1)}} \ge q_i$, then $f(T^{i-1} \cup \{i\}) \le \nu^\star/5$.

Let $D \subseteq X \cap A$ be a demand set with respect to the prices defined in
\Cref{alg:small_agents_SM}, $q_i=\frac{2}{\sqrt5}\sqrt{c_i \cdot \nu^\star}$.
We show that $f(T)$ is at least the net utility induced by $D$,
i.e., $f(D)-\sum_{i \in D}q_i$.

Observe that $\marg{i}{T^{(i-1)}} < q_i$ for any $i \in D \setminus T$, or we
would have $i \in T$ by assumption.
\begin{align*}
    f(D) - f(T)
    &\le f(D \cup T) - f(T) && (\text{monotonicity})\\
    &= f(D \mid T)\\
    &\le \sum_{i \in D \setminus T} \marg{i}{T} && (\text{submodularity})\\
    &\le \sum_{i \in D \setminus T} \marg{i}{T^{(i-1)}} && (\text{submodularity})\\
    &< \sum_{i \in D \setminus T} q_i && (i \notin T \text{ for every } i \in D \setminus T),\\
    &\le \sum_{i \in D} q_i && (q_i \ge 0).
\end{align*}
Under the good event we can bound $f(T)$ in terms of $\mu^\star$.  First, incentivizing the set $\smallAgents^\star$ requires payments at most $1/2$. Thus, we can apply
\Cref{lem:budgeted_sqrt_cost} with $S=\smallAgents^\star$ and $\rho=1/2$,
\begin{equation}\label{eq:sqrt_cost_Xstar}
\sum_{i\in\smallAgents^\star}\sqrt{c_i}\;\le\;\sqrt{\tfrac12 f(\smallAgents^\star)}.
\end{equation}
Second, the good event $f(\smallAgents^\star)\ge\frac58\nu^\star$ rearranges to
$\nu^\star \le \frac85 f(\smallAgents^\star)$.  Therefore
\begin{align*}
    f(T)
    &\ge
    f(D) - \sum_{i \in D} q_i \\
    &\ge
    f(\smallAgents^\star) - \sum_{i \in \smallAgents^\star} q_i  && (D \text{ is a demand set})\\
    &=
    f(\smallAgents^\star) - \frac{2}{\sqrt5}\sqrt{\nu^\star}\cdot \sum_{i \in \smallAgents^\star} \sqrt{c_i}  \\
    &\ge
    f(\smallAgents^\star) - \frac{2}{\sqrt5}\cdot\sqrt{\tfrac85 f(\smallAgents^\star)}\cdot \sqrt{\tfrac12 f(\smallAgents^\star)}
    && (\text{good event and \eqref{eq:sqrt_cost_Xstar}}) \\
    &=
    f(\smallAgents^\star)\bigg(1-\frac{2}{\sqrt5}\cdot\sqrt{\frac{8}{10}}\bigg)
    =
    f(\smallAgents^\star)\bigg(1-\frac{2}{\sqrt5}\cdot\frac{2}{\sqrt5}\bigg)
    =
    \frac{1}{5}f(\smallAgents^\star) \\
    &\ge
    \frac{1}{5}\cdot\frac58\nu^\star
    =
    \frac{\nu^\star}{8} && (f(X^\star) \ge \tfrac58\nu^\star)\\
    &\ge
    \frac{13}{14}\cdot\frac{\mu^\star}{8}
    =
    \frac{13}{112}\mu^\star && (\nu^\star \ge \tfrac{13}{14}\mu^\star).
\end{align*}

\medskip\noindent
\textbf{Case 2:} There exists an agent $i \in \smallAgents \cap A$ such that
$\marg{i}{T^{i-1}} \ge q_i$ and $\nu^\star/5< f(T^{i-1}\cup \{i\})$.
\begin{align*}
f(T) &\overset{(1)}{\ge}
f(T^{i-1})
\overset{(2)}{\ge} f(T^{i-1} \cup \{i\})-f(\{i\})
> \frac{\nu^\star}{5} - f(\{i\})
\ge \frac{\nu^\star}{5} - \max_{i \in \smallAgents} f(\{i\}) \\
&\overset{(3)}{\ge} \frac{\nu^\star}{5} - \frac{\mu^\star}{14}
\ge
\frac{13\mu^\star}{70} - \frac{5\mu^\star}{70}
=
\frac{4}{35}\mu^\star,
\end{align*}
where inequality (1) holds by the monotonicity of $f$, inequality (2) by subadditivity of $f$, and inequality (3) since $f_i \le \mu^\star/14$ for every
$i \in \smallAgents$, and the last since $\nu ^\star \ge \frac{13}{14}\mu^\star$.

\medskip\noindent
Since $\frac{13}{112} = \frac{65}{560} > \frac{64}{560} = \frac{4}{35}$, the
binding case is Case 2, and $f(T)\ge\frac{4}{35}\mu^\star$ in both.  This
concludes the proof.
\end{proof}

We are now ready to prove the main result of this section.  We use the fact that
the good event occurs with probability at least $117/181$, which is established
in \Cref{lem:small_agents_good_case}, see \Cref{sec:good_event}.

\begin{proof}[Proof of \Cref{prop:muStar_smallAgentCase}]
First, \Cref{alg:small_agents_SM} always outputs a budget-feasible contract, as
shown in \Cref{lem:smallAgents_bounded_payment}; in particular the principal
retains at least half of the reward.  To bound the expected
utility, consider a set of active agents $A$, let $T \subseteq A$ be the set of
agents to receive non-zero contracts by \Cref{alg:small_agents_SM}, and let
$W \subseteq T$ be the set of effort-exerting agents under the best equilibrium
induced by these payments.  Conditioning on the good event $\mathcal{G}$,
\[
\E \bigg[\bigg(1-\sum_{i \in A}\alpha_i\bigg) f(W) \;\bigg|\; \mathcal{G} \bigg]
    \ge
    \frac12 \E[f(W) \mid \mathcal{G}]
    \ge
    \frac14 \E[f(T) \mid \mathcal{G}]
    \ge
    \frac14\cdot\frac{4}{35} \mu^\star
    =
    \frac{1}{35} \mu^\star,
\]
where the first inequality follows from \Cref{lem:smallAgents_bounded_payment},
the second from \Cref{lem:smallAgents_lowerbound_fW}, and the third from
\Cref{lem:smallAgents_lowerbound_fT}.  By the non-negativity of $f$ we conclude
that
\[
\E \bigg[\bigg(1-\sum_{i \in A}\alpha_i\bigg) f(W)\bigg]
\ge
\Pr[\mathcal{G}]\cdot \E\bigg[\bigg(1-\sum_{i \in A}\alpha_i\bigg) f(W) \;\bigg|\; \mathcal{G} \bigg]
\ge
\frac{117}{181}\cdot\frac{1}{35}\,\mu^\star
=
\frac{117}{6335}\,\mu^\star,
\]
where the last inequality follows from \Cref{lem:small_agents_good_case}.  Since
$\frac{117}{6335} \ge \frac{1}{55}$ (indeed $6335/117 = 54.14\ldots$), the claim
follows.
\end{proof}

\begin{proof}[Proof of \Cref{thm:SM_const_CR}]
By \Cref{prop:muStar_largeAgentCase} and \Cref{prop:muStar_smallAgentCase}, the
expected principal's utility is at least
$\min\{\frac{1}{56},\frac{117}{6335}\}\,\mu^\star = \frac{1}{56}\mu^\star$, so
\Cref{alg:small_agents_SM} together with the single-agent policy is
$56$-competitive for the online budgeted problem. By
\Cref{prop:reduction_to_budget} this implies a competitive ratio of
$3(56+2)=174$ for the online contracting problem.
\end{proof}

\subsection{Lower Bounding the Good Event}\label{sec:good_event}
\newcommand{\RVx}{\mathcal{X}}

Recall that $\mu^\star = \E[f(\cheapAgents^\star)]$, where $L^\star$ is the
optimal set for the contracting problem with budget $1/2$ and agents in
$L = \{i \in A \mid c_i/f_i \le 1/2\}$ (see \Cref{eq:budget_half}), and that
$\nu^\star = \E[f(\smallAgents^\star)]$ is the optimal solution to the problem
when restricting attention to the agents in
$X = \{ i \in L \mid f_i \le \mu^\star/14\}$ (see \Cref{eq:_Xstar}).  We show
that whenever $\nu^\star / \mu^\star \ge 13/14$, with probability at least
$117/181$ it holds that $f(\smallAgents^\star) \ge \frac58\nu^\star$.

To this end, let $h:\smallAgents \to \reals_{\ge0}$ be the function that maps the
set of active agents in $\smallAgents$ to the value of $f(\smallAgents^\star)$,
as specified in \Cref{eq:_Xstar}.  Namely, for any
$\smallAgents' \subseteq \smallAgents$,
\begin{equation}\label{eq:def_h}
    h(\smallAgents') = \max_{S \subseteq \smallAgents'} f(S) \quad \text{s.t.} \quad \sum_{i \in S} \frac{c_i}{\marg{i}{S}} \le \frac{1}{2}.
\end{equation}

\begin{restatable}{claim}{hisXOS}\label{clm:h_is_XOS}
Whenever $f$ is submodular, $h$ is XOS.
\end{restatable}
The proof is deferred to \Cref{sec:h_is_XOS_proof}.

We use the fact that $h$ is XOS to apply the Efron--Stein inequality, and bound
the variance of $h$ in terms of the maximal marginal contribution of agents in
$X$.  By definition, agents in $X$ have small marginal contributions, at most
$\mu^\star/14$, which implies, together with
$\nu^\star \ge \frac{13}{14}\mu^\star$, that the variance of $h$ can be upper
bounded by $\frac{1}{13}(\nu^\star)^2$.  Applying Cantelli's inequality then
concludes the proof.

\begin{lemma}\label{lem:small_agents_good_case}
Assume $\nu^\star \ge \frac{13}{14}\mu^\star$.  Then
$\Pr\big[f(\smallAgents^\star) > \frac{5}{8}\nu^\star\big]\ge \frac{117}{181} > 0.646$.
\end{lemma}
\begin{proof}
Given an instance of the problem, we assume without loss of generality that the
agents in $\smallAgents$ are indexed $1, \dots, |\smallAgents|$.  Let $\RVx_i$ be
the random variable that indicates whether $i$ is active.  Slightly abusing
notation, we define $h(\RVx_1,\dots,\RVx_{|\smallAgents|}) =h(A(\RVx))$, where
$A(\RVx) = \{i \in \smallAgents \mid \RVx_i=1\}$.  The Efron--Stein inequality
\cite{efron1981jackknife} gives
\[
\Var(h(\RVx)) \le \frac12 \sum_{i=1}^n\E[(h(\RVx)-h(\RVx^{(i)}))^2],
\]
where
$\RVx^{(i)} = (\RVx_1,\dots,\RVx_{i-1},\RVx'_i,\RVx_{i+1},\dots,\RVx_{|\smallAgents|})$
and $\RVx'_i$ is an i.i.d.\ copy of $\RVx_i$.

First, observe that if $\RVx_i = \RVx'_i$ then $|h(\RVx)-h(\RVx^{(i)})| = 0$, and
the probability that $\RVx_i \ne \RVx'_i$ is exactly $2p_i(1-p_i)$, for any
$\RVx_{-i}$.

\begin{observation}\label{obs:h_diff}
Assume $f$ is subadditive and $\RVx_i \ne \RVx'_i$.  Then
$|h(\RVx)-h(\RVx^{(i)})| \le \mu^\star/14$.
\end{observation}
\begin{proof}
Observe that $h$ is monotone non-decreasing, so if $\RVx_i \ne \RVx'_i$ then
$|h(\RVx)-h(\RVx^{(i)})| = h(1,\RVx_{-i}) - h(0,\RVx_{-i})$.  Let $S$ be the set
that attains the value $h(1,\RVx_{-i})$.  If $i \notin S$ then
$h(1,\RVx_{-i}) = h(0,\RVx_{-i})$.  Otherwise, by subadditivity,
$h(1,\RVx_{-i}) = f(S) \le f(\{i\}) + f(S \setminus \{i\}) \le \mu^\star/14 + h(0,\RVx_{-i})$,
and the claim follows.
\end{proof}

\begin{align*}
\Var(h(\RVx))
&\le
\frac12\sum_{i=1}^{|\smallAgents|}\E_{\RVx,\RVx'_i}[(h(\RVx)-h(\RVx^{(i)}))^2] && \text{(Efron-Stein)}\\
    &=
    \sum_{i=1}^{|\smallAgents|}p_i(1-p_i)\E_{\RVx_{-i}}[(h(1,\RVx_{-i})-h(0,\RVx_{-i}))^2] \\
    &\le
    \frac{\mu^\star}{14}
    \sum_{i=1}^{|\smallAgents|}p_i\cdot \E_{\RVx_{-i}}[h(1,\RVx_{-i})-h(0,\RVx_{-i})] && (\text{\Cref{obs:h_diff}})\\
    &\le
    \frac{\mu^\star}{14}\,\E\big[h(A(\RVx))\big] && (\text{\Cref{lem:expected_XOS_margs}})\\
     &=
    \frac{\mu^\star}{14} \cdot \nu^\star && (\nu^\star = \E_A[f(\smallAgents^\star)]) \\
    &\le
    \frac{1}{14}\cdot\frac{14}{13}\nu^\star\cdot\nu^\star
    =
    \frac{(\nu^\star)^2}{13}  && (\nu^\star \ge \tfrac{13}{14}\mu^\star).
\end{align*}

Write $\sigma^2 = \Var(h(\RVx)) \le \frac{1}{13}(\nu^\star)^2$.  By Cantelli's
inequality (\Cref{fact:cantelli}) applied to $Z = f(\smallAgents^\star)$ with
mean $\nu^\star$ and $t = \frac38\nu^\star$,
\[
\Pr\bigg[f(\smallAgents^\star) \le \frac58\nu^\star\bigg]
=
\Pr\bigg[f(\smallAgents^\star) \le \nu^\star - \frac38\nu^\star\bigg]
\le
\frac{\sigma^2}{\sigma^2 + \frac{9}{64}(\nu^\star)^2}
\le
\frac{\frac{1}{13}}{\frac{1}{13} + \frac{9}{64}}
=
\frac{64}{64+117}
=
\frac{64}{181},
\]
where the second inequality uses that $x \mapsto \frac{x}{x+c}$ is increasing.
We conclude that 
$\Pr\big[f(\smallAgents^\star) > \frac58\nu^\star\big] \ge 1 - \frac{64}{181} = \frac{117}{181}$, and the claim follows.
\end{proof}

\section{\boldmath Safe Policies Incur Super-Constant Loss for Submodular $f$}
In this section we show a super-constant gap between safe and general policies when $f$ is submodular. Our construction is based on \cite{gobel2014online}, and implies a stronger result:\footnote{Similar constructions have also appeared in \cite{im2011secretary,chawla2019pricing}.}
The gap between safe and general policies persists even when the arrival order is random: At each time step an agent is selected uniformly at random and the algorithm is informed whether this agent is active or not.

\begin{theorem}\label{thm:submodular_gap}

    There exists an online contract instance with coverage $f$ such that the ratio between the expected utilities of the optimal general policy and the optimal safe policy is $\Omega\left(\frac{\log n}{(\log \log n)^2}\right)$.
    This holds even if agents arrive in random order, rather than adversarial.
    
\end{theorem}

\paragraph{Construction.}
We consider a laminar instance. We set
\begin{equation}\label{eq:d_ary}
    d = \left\lceil \max\left\{ 4 \left(2 \frac{\log n}{\log \log n} + 2\right)^2 + 1, \sqrt{\log n} \right\} \right\rceil 
\end{equation}
and consider a complete $d$-ary tree of height $h$. We choose $h$ to be the smallest integer so that the tree has at least $n$ nodes. Note that
\[
\left\lceil \frac{\log n}{\log d} \right\rceil - 1 \leq h \leq \left\lceil \frac{\log n}{\log d} \right\rceil + 1,
\]
If necessary, we remove leaves from the tree so that it has exactly $n$ nodes. This way, each path from the root to a leaf has length either $h-1$ or $h$.

Agents correspond to nodes in the tree, where all have the same positive, and arbitrarily small cost $c_i = \eps>0$ for all $i \in N$. Agents also share the activation probabilities, for every $i \in N$ $p_i = p = \frac{1}{2h}$.

We construct a coverage function $f$ defined with respect to a universe $U$.
We have one element in $U$ for each of the leaves of the tree.
The elements that can be covered are the leaves of the tree. Each agent covers the set of leaves it is an ancestor of. In particular, an agent which corresponds to a leaf only covers that leaf, and the root of the tree covers all elements.
Our construction is illustrated in \Cref{fig:submodular-gap}.

Since costs are arbitrarily small, the payments to the agents made by any reasonable policy are negligible, which implies that expected principal's utility $\approx$ expected reward (i.e., $f$ value).
Thus, in the rest of the section we focus on bounding the expected reward generated by a given policy.

\begin{figure}[htbp]
\centering
\begin{tikzpicture}[
    x=1cm,y=1cm,
    every node/.style={font=\small},
    agent/.style={circle, draw, minimum size=5mm, inner sep=0pt},
    active/.style={circle, draw, fill=gray!35, minimum size=5mm, inner sep=0pt},
    colbox/.style={draw, rounded corners, inner sep=6pt},
    itembox/.style={rectangle, draw, fill=white, minimum size=5mm, inner sep=0pt},
    level 1/.style={sibling distance=4.8cm, level distance=1.2cm},
    level 2/.style={sibling distance=2.4cm, level distance=1.2cm},
    level 3/.style={sibling distance=1.2cm, level distance=1.2cm},
    edge from parent/.style={draw, thick}
]

\node[agent] (root) at (4, 4.5) {}
    child {node[active] (l) {}
        child {node[agent] (ll) {}
            child {node[active] (lll) {}}
            child {node[agent] (llr) {}}
        }
        child {node[active] (lr) {}
            child {node[agent] (lrl) {}}
            child {node[active] (lrr) {}}
        }
    }
    child {node[agent] (r) {}
        child {node[active] (rl) {}
            child {node[agent] (rll) {}}
            child {node[agent] (rlr) {}}
        }
        child {node[agent] (rr) {}
            child {node[active] (rrl) {}}
            child {node[agent] (rrr) {}}
        }
    };

\foreach \n/\i in {lll/1, llr/2, lrl/3, lrr/4, rll/5, rlr/6, rrl/7, rrr/8} {
    \node[itembox] (u\i) at ([yshift=-1cm]\n) {$u_{\i}$};
    \draw[dashed, gray] (\n) -- (u\i);
}


\node[align=right, anchor=east] at (-1.5, 4.5) {Level 0\\(Root)};
\node[align=right, anchor=east] at (-1.5, 3.3) {Level 1};
\node[align=right, anchor=east] at (-1.5, 2.1) {$\vdots$};
\node[align=right, anchor=east] at (-1.5, 0.9) {Level $K$\\ (Leaves)};
\node[align=right, anchor=east] at (-1.5, -0.1) {Items};

\draw[decorate,decoration={brace,amplitude=6pt,mirror}]
    ([xshift=-2mm,yshift=-2mm]u1.south west) -- ([xshift=2mm,yshift=-2mm]u8.south east)
    node[midway,below=8pt] {Universe $U$ of $\Theta(n)$ elements};

\node[draw, rounded corners, inner sep=4pt, dashed, thick, fit=(u5)(u6)(u7)(u8)] (covL) {};
\draw[dashed, thick, ->] (r) to[bend left=30] 
    node[midway, right=4pt, align=left, font=\footnotesize] {An agent covers\\all items in\\its subtree} (covL.north east);
\end{tikzpicture}
\caption{Visualization of the hierarchical submodular coverage gap instance, for $d=2$. 
A safe policy must rigidly select an anti-chain to prevent coverage overlap, while a general policy organically utilizes ancestral dominance.}
\label{fig:submodular-gap}
\end{figure}
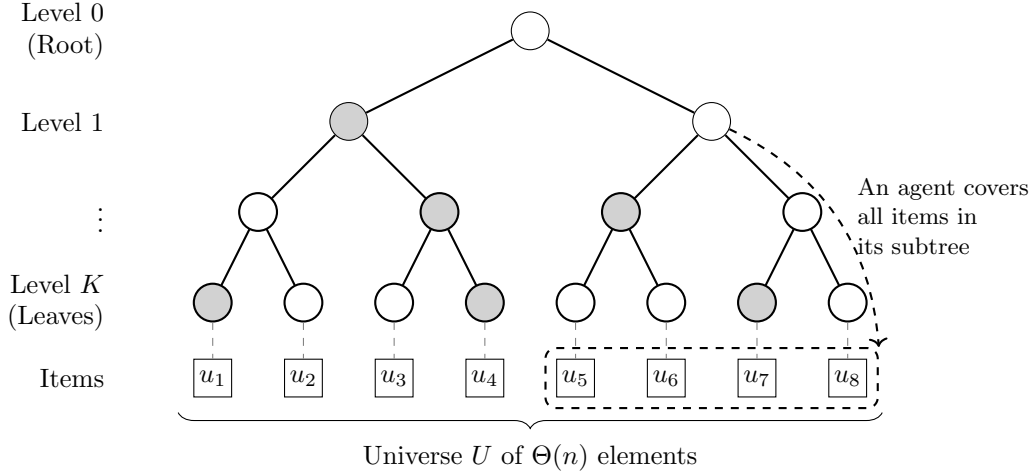

\begin{observation}\label{obs:unsafe_SM_gap}
    There exists an online policy which achieves an expected utility of $\Theta(|U|)$.
\end{observation}
\begin{proof}
    It is a feasible online policy to accept all active agents (regardless of the order in which they are revealed), and offer each a contract of $\eps$. In every equilibrium, exactly those agents will work who have no active ancestor in the tree. 
    The expected principal utility will be (up to the negligible payments to the agents) the expected number of leaves with the property that at least one node on the path to the root is active. 
    For each leaf, the probability is at least $1-(1 - \frac{1}{2h})^h \geq 1 - \frac{1}{\sqrt{e}}$, 
    and the expected utility is at least $(1 - \frac{1}{\sqrt{e}}) |U|$.
\end{proof}

\paragraph{Safe Policies.}

Recall that in a safe policy, the principal selects a subset of active agents $S \subseteq A$, in an online fashion, and pays each agent $c_i/\marg{i}{S}$ at the end of the process.
In particular, any optimal safe policy, for any realized set of active agents, selects a set $S$ such that every $i \in S$ has a strictly positive marginal, i.e., $\marg{i}{S} > 0$. Otherwise, the principal will suffer infinite payment.
In the context of our construction, the above implies the observation below.
\begin{observation}\label{obs:safe_no_anc}
    Any optimal online policy, for any realized set of active agents, does not select two agents $u$ and $v$ such that $u$ is an ancestor of $v$.
\end{observation}

It has been shown on \cite{gobel2014online} that the following policy, which we denote by $\pi$, is optimal for an online selection problem:
Accept an agent if and only if all of its ancestors in the tree cannot be accepted, i.e., every ancestor is either inactive, or it is an ancestor of an already accepted agent.
Below we prove an analogous result for the problem of online contracting, where an accepted agent receives a contract of $c_i/|U(i)|$, where $|U(i)|$ is the number of leaves in the subtree rooted at $i$ (i.e., the number of elements covered by $i$).
\begin{proposition}[\cite{gobel2014online}]\label{prop:pi_is_opt_safe}
    The policy $\pi$ is the optimal safe policy.
\end{proposition}
\begin{proof}
    We prove the claim by induction.
    Consider the arrival of agent $i$, which can be selected by a safe policy, i.e., no descendant of $i$ has been selected. 
    Let $U(i) \subseteq U$ be the set of elements covered by $i$, and let $T_i$ be the sub-tree of agents rooted at $i$.
    By induction hypothesis, we know that applying $\pi$ in future steps is optimal. We now examine whether $i$ should be selected at the current step.
    
    If $i$ has no ancestor that could be accepted, taking $i$ is optimal -- there is no better way to cover $U(i)$ or a subset thereof. This would coincide with $\pi$.
    
    Otherwise, let $j$ be an ancestor of $i$ that could be accepted, which is closest to the root of the tree.
    Let $U(j) \supsetneq U(i)$ be the set of elements $j$ covers, and let $T_j$ be the sub-tree rooted at $j$.
    First, $U(j)$ is completely uncovered, or $j$ could not have been chosen by a safe policy. Moreover, because the tree is $d$-ary, $|U(j)| \ge (d-1)|U(i)|$.
     
    By \Cref{obs:safe_no_anc}, any safe policy which selects $i$ cannot select $j$. 
    Such a policy will gain $U(i)$ with probability $1$, we now compute a lower bound on its expected loss by not picking agent $j$.
    Let $\pi'$ be the policy which selects $i$ and then continues by applying $\pi$. It is enough to show that $\pi'$ will yield lower expected utility than rejecting $i$ and applying $\pi$.

    Since $j$ is the closest-to-root ancestor of $i$ that can be selected, $\pi'$ and $\pi$ behave identically when considering future agents outside $T_j$.
    Observe that $\pi'$ and $\pi$ will also behave identically for agents in $T_j \setminus \{j\}$, except for those in $T_i \setminus \{i\}$, which may be selected by $\pi$ and are not available for $\pi'$. We will show that $\pi'$ yields less than $\pi$, even when ignoring those agents.
    The above implies two things:
    \begin{itemize}
        \item If $j$ is not active, $\pi'$ covers at most $|U(i)|$ more elements than $\pi$.
        \item If $j$ is active, $\pi$ selects it, and thus cover any elements in $U(j)$. $\pi'$ cannot select $j$.
    \end{itemize}

     Now we bound the difference in expected reward when $j$ is active.
     To this end, for any element $u \in U(j) \setminus U(i)$, we define the event $X_{u}$ in which, besides $j$, every yet-to-be-revealed ancestor of $u$ is inactive. 
     Observe that (i) if $X_u$ holds, than $u$ cannot be covered by $\pi'$, (ii) $\Pr[X_u]\ge(1-p)^h$, and (iii) any such event is independent of whether $j$ is active.
    \begin{align*}
    \E[(\pi - \pi') \cdot \textbf{1}[j \text{ is active}]]
    &=
    p \cdot \sum_{u \in U(j) \setminus U(i)} (\Pr[\pi \text{ covers }u \mid j \text{ is active}] - \Pr[\pi' \text{ covers }u \mid j \text{ is active}]) \\
    &=
    p \cdot \sum_{u \in U(j) \setminus U(i)} (1 - \Pr[\pi' \text{ covers }u \mid j \text{ is active}]) \\
    &\ge
    p \cdot \sum_{u \in U(j) \setminus U(i)} \Pr[X_u] \cdot (1 - \Pr[\pi' \text{ covers }u \mid j \text{ is active} \land X_u]) \\
     &\ge
    p \cdot \sum_{u \in U(j) \setminus U(i)} (1-p)^h \cdot (1 - 0) \hspace{4.6cm} (\text{(i) and (ii)}) \\
    &\ge
    p(1-p)^h (|U(j)|-|U(i)|).
    \end{align*}

    Using the fact that $|U(j)| \ge (d-1)|U(i)|$, we get
    \begin{align*}
    \E[(\pi-\pi')]
    &=
    \E[
    (\pi-\pi') \cdot \textbf{1}[j\text{ is active}]
    +
    (\pi-\pi') \cdot \textbf{1}[j\text{ is not active}]]
    \ge
    p(1-p)^h (d-2)|U(i)|
    \end{align*}

    Finally, we use the definitions of $d,h$ and $p$ to show that the above is non-negative.
    First, by Bernoulli's inequality, $(1-p)^h \ge 1-ph = 1/2$.
    Since $d \ge \sqrt{\log n}$, we have 
    \[
    h \le \left\lceil\frac{\log n}{\log d}\right\rceil+1 
    \le 
    \left\lceil\frac{2\log n}{\log \log n}\right\rceil+1 \le \frac{2\log n}{\log \log n}+2
    \le \sqrt{\frac{d-1}{4}},
    \]
    where the last inequality follows from \Cref{eq:d_ary}.
    Combining the above together with $p = 1/(2h)$ we get
    \[
    \E[\pi - \pi'] \ge \frac{d-2}{4h}|U(i)| \ge \frac{2(d-2)}{\sqrt{d-1}} \ge 0,
    \]
    where the last inequality holds for any $d \ge 2$.
\end{proof}

\begin{proposition}\label{prop:pi_util_bound}
    The policy $\pi$ generates expected utility of $ O\left(|U| \cdot \frac{(\log \log n)^2}{\log n}\right)$.
\end{proposition}
\begin{proof}
    Consider node $i$ with height $\ell \in \{0,\dots,h\}$.
    Such a node will be selected by $\pi$ if and only if it is active and every ancestor of $i$ has already appeared and was inactive.
    Since $i$ has $\ell$ ancestors, the probability of this event is
     $$\Pr[\pi \text{ selects agent $i$ of height $\ell$}] \le p \cdot \frac{1}{\ell+1}(1-p)^{\ell} \le p/(\ell+1)$$.
    
    The probability that an element $u \in U$ is covered by $\pi$ is the probability that one of its ancestors is selected.
    \[
    \Pr[\pi \text{ covers } u] \le \sum_{\ell=0}^h \Pr[\pi \text{ selects $u$'s ancestor of height $\ell$}]
    \le \sum_{\ell=0}^h \frac{p}{\ell+1} = O(p \log h) = O\left(\frac{\log h}{h}\right),
    \]
    where the last inequality follows since $p = 1/(2h)$.

Recall that $h = O\left(\frac{\log n}{\log \log n}\right)$.
The expected utility generated by $\pi$ is thus
\[
\E[\pi] = \sum_{u \in U} \Pr[\pi \text{ covers } u] = O\left(|U| \cdot \frac{(\log \log n)^2}{\log n}\right)
\]

\end{proof}

The proof of \Cref{thm:submodular_gap} follows immediately from the above results.
\begin{proof}[Proof of \Cref{thm:submodular_gap}]
By \Cref{obs:unsafe_SM_gap}, an online policy can achieve expected utility of $\Omega(|U|)$. \Cref{prop:pi_is_opt_safe} and \Cref{prop:pi_util_bound} together imply that a safe policy achieves utility $O\left(|U|\frac{(\log \log n)^2}{\log n}\right)$.
\end{proof}

\section{\boldmath Super-Constant Competitive Ratio for XOS $f$}
\label{sec:xos}

In this section, we show that the constant-factor approximation achievable for submodular reward functions cannot be extended to the broader class of XOS functions. 
Moreover, there is a super-constant gap between the optimal (not necessarily \safe) online policy and the best \safe online policy.
\begin{theorem}\label{thm:xos}
    There exists an instance with $n$ agents and an XOS reward function $f$ such that,
    \begin{itemize}
        \item The expected utility of the prophet is $\Theta\big(\frac{\log n}{\log \log n}\big)$.
        \item The expected utility of the best (not necessarily \safe) online policy is $\Theta \big(\frac{\log \log n}{\log \log \log n}\big)$.
        \item The expected utility of any safe policy is $O(1)$.
    \end{itemize}
\end{theorem}
Our construction adapts the instance presented in \cite{babaioff2007matroids}, which is phrased in terms of online selection,
to the richer online contracting setting.
We consider a $\sqrt{n} \times \sqrt{n}$ grid of identical agents, each agent $i$ has cost $c_i = \frac{1}{2m}$, for some $m$ to be determined later, and is active with probability $p_i = \frac{1}{\sqrt{n}}$.
The agents are partitioned into $\sqrt{n}$ columns, $C_1, \dots, C_{\sqrt{n}}$, and the reward from a set of agents $S$ is determined by the largest intersection with a specific column, i.e., $f$ is the following XOS function $f(S) = \max_{j \in [\sqrt{n}]} |S \cap C_j|$. Do note that $f$ is not submodular.

We describe the optimal contracts and the equilibria in this instance.
Consider some set $S$ of active agents and observe that for any $i \in S$, we have $\marg{i}{S} \in \{0,1\}$:
it is $1$ if $i$ belongs to the unique column with the largest intersection with $S$, and zero otherwise.
Since the costs of all agents is $c = 1/(2m)$, any optimal policy only offers contracts of the form $\alpha_i \in \{0,1/(2m)\}$.
The structure of the equilibrium induced by $\contract$ is as follows:
Let $S = \{i \in A \mid \alpha_i \ge 1/(2m)\}$, and let $j^\star \in \argmax_{j \in [\sqrt{n}]}|S \cap C_j|$ (ties broken arbitrarily). The set of agents $S \cap C_{j^\star}$ is the equilibrium that maximizes $f$, i.e., $S(\contract) = S \cap C_{j^\star}$.
Indeed, any agent $i \in S \cap C_{j^\star}$ has a marginal $\marg{i}{S \cap C_{j^\star}}=1$, and $\alpha_i = c_i / \marg{i}{S \cap C_{j^\star}}$.
On the other hand, for every $i \notin S \cap C_{j^\star}$ we have $\marg{i}{S \cap C_{j^\star}}=0$, and $\alpha_i < c_i / \marg{i}{S \cap C_{j^\star}}$.

\paragraph{The Offline Benchmark.}
Analyzing the offline optimum is  similar to the analysis of \cite{babaioff2007matroids}.
Let $A \subseteq N$ denote the realized set of active agents. For each column $C_j$, the number of active agents $X_j = |A \cap C_j|$ follows a binomial distribution $\mathrm{Bin}(\sqrt{n}, 1/\sqrt{n})$, and the random variables
\(\{X_j\}_{j\in[\sqrt n]}\) are independent. 
A maximum-load argument, similar to the classical balls-into-bins process \cite{babaioff2007matroids,RaabSteger1998BallsIntoBins}
then gives $\E\left[\max_{j\in[\sqrt n]} X_j\right] = \Theta\left(\frac{\log n}{\log\log n}\right)$; we give a self-contained proof, with the explicit leading constant, in \Cref{sec:max-load-proof} (\Cref{prop:max-load-n}):
\[
\mathbb{E}\left[\max_{j\in[\sqrt n]} X_j\right]
=
\left(\frac12+o(1)\right)\frac{\log n}{\log\log n}.
\]
We set $m=\mathbb{E}\left[\max_{j\in[\sqrt n]} X_j\right]$, so the instance is now completely defined.

\begin{proposition}
    \label{prop:xos-prophet-lower-bound}
    The expected offline optimum is $\Omega(\frac{\log n}{\log \log n})$.
\end{proposition}

\begin{proof}
Consider the offline algorithm which, given in advance the set $A$ of active agents, offers $\alpha_i=1/(2m)$ to $\min\{m,|A \cap C_{j^\star}|\}$ agents that belong to $C_{j^\star}$, where $j^\star \in \arg\max_{j \in \sqrt{n}} X_j$.
The marginal contribution of each agent with $\alpha_i=1/(2m)$ is $1$, so $\alpha_i = c_i/\marg{i}{S}$ and the agent exerts effort in equilibrium.
The sum of payments is now at most $m \cdot \frac{1}{2m} = \frac12$, so the principal retains at least a $\frac12$ fraction of the reward, and the expected performance of this algorithm is at least $\frac{1}{2} \cdot \E[\min(X_{j^\star},K)]$.

In order to bound $\E[\min(X_{j^\star},K)]$, let $Y_j = 1$ if $X_j \ge m$, and recall that $m = \frac{\ln \sqrt{n}}{3 \ln \ln \sqrt{n}}$. 
\[
\Pr[Y_j = 1] \ge \Pr[X_j = m] = \binom{\sqrt{n}}{m} \left( \frac{1}{\sqrt{n}} \right)^m \left( 1 - \frac{1}{\sqrt{n}} \right)^{n-m} \geq \left(\frac{\sqrt{n}}{m}\right)^m \frac{1}{\sqrt{n}^m} \frac{1}{e} = \frac{1}{m^m e} = \frac{1}{e n^{1/6}}.
\]
So $\mu := \sum_{j \in [\sqrt{n}]} \Pr[Y_j = 1] \geq \frac{n^{1/3}}{e}$. As the $Y_j$ are independent and their variance is bounded by their expectation, Chebyshev's inequality gives us
\[
\Pr\left[\sum_{j \in [\sqrt{n}]} Y_j = 0\right] \leq \Pr\left[\left\lvert\sum_{j \in [\sqrt{n}]} Y_j - \mu \right\rvert \geq \mu\right] \leq \frac{\mathrm{Var}(\sum_{j \in [\sqrt{n}]} Y_j)}{\mu^2} \leq \frac{1}{\mu} \leq \frac{e}{n^{1/3}}.
\]

Consequently,
\[
\E[\max_j \min\{K, X_j\}] \geq \left(1 - \frac{e}{n^{1/3}} \right) m = \Omega\left(\frac{\log n}{\log \log n}\right).
\]
This concludes the proof.
\end{proof}
    
\subsection{Safe Policies}
Note that for a policy to be safe all hired agents must belong to the same column. Roughly speaking, the independence between the different agents implies that any safe policy has to ``guess'' a column and hires all agents from that column.
This argument also appears in \cite{babaioff2007matroids}, and we provide it here for completeness.

\begin{proposition}[Safe Policies Yield $O(1)$ Utility] \label{prop:safe-o1}
    Any safe online policy yields an expected principal utility of at most $O(1)$.
\end{proposition}

\begin{proof}
    Consider the optimal safe policy and let $T$ denote the (random) set of agents that receive a non-zero contract.
    By the definition of a safe policy, for every realization of $T$, all agents in $T$ exert effort. In particular, every agent in $T$ has strictly positive marginal utility, or the principal will suffer infinite payment. Thus, for any realized set $A$ of active agents, there exists some column $C_j$ such that $T \subseteq A \cap C_j$.

    If the policy never makes any offers, its utility is $0$. 
    Otherwise, let $t$ be the first time step in which the algorithm offers a non-zero contract, $\alpha_t > 0$ 
    to agent $t$. Let $C^\star$ denote the column that agent $t$ belongs to. 
    Since the policy is safe, all  agents which eventually receive a non-zero contract 
    must belong to $C^\star$.

     Conditioned on this, let 
     $Y_t \subseteq C^\star \setminus \{t\}$
     be the set of agents in 
     $C^\star$ which did not arrive by time $t$. 
     Every agent in $Y_t$ is active with probability $1/\sqrt{n}$ independent of all other agents.
     Thus, the expected number of active agents from $C^\star$ 
     arriving after $t$ is given by $|Y_t|/\sqrt{n} \le 1$.
     Since each active agent in 
     $C^\star$ contributes at most $1$ to the principal's utility, we have that the expected utility, together with agent $t$, 
     is at most $2$.
     The above holds irrespective of how $t$ is chosen, and the claim follows.
\end{proof}

\subsection{General Policies}

We turn to analyze general policies, which may hire agents from different columns. Here, we depart from the selection problem of \cite{babaioff2007matroids}, in which any algorithm can only select elements from a single column.
We utilize two new insights to establish a separation between the offline optimum, the optimal online policy and the optimal \safe policy.
We first upper bound the performance of any online policy, by observing that the principal can incentivize at most $2m$ agents, as incentivizing one costs $1/(2m)$. This roughly implies that an online policy can achieve the expected maximum among $2m=\Theta(\log n/ \log \log n)$ columns, compare to the $\sqrt{n}$ columns in the offline optimum.
Then we characterize the optimal online policy for this instance (up to a constant factor): Pick $m$ columns at random and offer $\alpha_i=1/(2m)$ to the first $3m/2$ agents in those columns. 
We show that even with the budget constraint, the policy achieves the expected maximum among $m$ columns with constant probability. As such, its expected utility is only a constant-factor away from the upper bound.

\begin{proposition} 
\label{prop:impossibility-xos-safe}
    Any online policy, safe or not, yields an expected utility of $O(\frac{\log \log n}{\log \log \log n})$. 
\end{proposition}

\begin{proof}
    First, observe that any agent which receives $\alpha_i <1/(2m)$ will not exert effort under any equilibrium, and any agent which is incentivized to work by $\alpha_i < 1$, can also be incentivized by setting $\alpha_i=1/(2m)$.
    Thus, we can assume without loss of generality that an online policy only offers contracts of the form $\alpha_i \in \{0, \frac{1}{2m}\}$.
    Thus, any optimal policy will, in any realization, offer non-zero contracts to at most $2m$ agents, and will span at most $2m$ different columns. 

    Assume the online policy has offered contracts to agents from $k \le 2m$ different columns, and denote them by $C_1,\dots, C_k$.
    Applying the same argument as in \Cref{prop:safe-o1}, for any $j \in [k]$, the number of active agents in $C_j$ yet to arrive is stochastically dominated by $Z_j \sim \mathrm{Bin}(\sqrt{n},1/\sqrt{n})$.
    
    Let $S$ be the agents that exert effort. 
    Since the number of queried columns is $k \leq 2m$, we can upper bound the principal's utility by $1$ plus the expected maximum of $2m$ independent Binomial random variables $Z_j \sim \mathrm{Bin}(\sqrt{n},1/\sqrt{n})$:
    \[
    1+\E\bigg[\max_{j \in [k]} Z_j\bigg] 
    = O\bigg(\frac{\log k}{\log \log k}\bigg) 
    = O\bigg(\frac{\log 2m}{\log \log 2m}\bigg) 
    = O\bigg(\frac{\log \log n}{\log \log \log n}\bigg).
    \]
    This concludes the proof.
\end{proof}

Below we show the tightness of the upper bound achieved in \Cref{prop:impossibility-xos-safe} and design an online policy that achieves expected utility of $\Omega\bigg(\frac{\log \log n}{\log \log \log n}\bigg)$. 

\begin{proposition}
    There exists an online policy which yields an expected utility of $\Omega\bigg(\frac{\log \log n}{\log \log \log n}\bigg)$.
\end{proposition}

\begin{proof}
    Consider an online policy which first picks $m$ columns uniformly at random, denote them $C_1,\dots,C_m$. 
    The policy offers a contract $\alpha_i = 1/(2m)$ to the first $3m/2$ agents which belong to $C_{[m]}=\cup_{j \in [m]} C_j$.
    As such, the total payment does not exceed $3/4$.

    Let $Z$ denote the number of active agents in $C_{[m]}$, it holds that $Z \sim \textrm{Bin}(m \sqrt{n},1/\sqrt{n})$, and in particular $\E[Z]=m$. 
    Let $M = \max_{j \in [m]} X_j$ be the maximal number of active agent in a single column among $C_1,\dots,C_m$. 
    Clearly, $M \le Z$.
    
    Let $\mathcal{E}$ be the ``good'' event in which $Z \le \frac{3}{2}m$. 
    Conditioning on the good event, the principal offers a contract to any active agent in $C_{[m]}$. Since the reward is determined by the column with the most active agents among $C_1,\dots,C_m$, the principal's utility under this event is at least $(1-3/4)\cdot M=M/4$.
    
    We can lower bound the principal's expected utility by 
    \[
    \frac14\E[M\cdot\mathbf{1}_{\mathcal{E}}] = \frac14\big(\E[M] - \E[M\cdot \mathbf{1}_{\neg\mathcal{E}}]\big)
    \]

    As $M$ is the maximum of $m$ i.i.d random variables distributed according to $\mathrm{Bin}(\sqrt{n},1/\sqrt{n})$, 
    \[
    \E[M] 
    = \Theta\bigg(\frac{\log m}{\log \log m}\bigg)
    = \Theta\bigg(\frac{\log \log n}{\log \log \log n}\bigg),
    \]
    where the last equality holds since $m = \frac{\log n}{\log \log n}$.
    Thus, it is enough to show that the error term, $\E[M\cdot \mathbf{1}_{\neg\mathcal{E}}]$, tends to zero as $n$ grows.

    We first use the fact that $M \le Z$. 
    \begin{align*}
        \E[M\cdot \mathbf{1}_{\neg\mathcal{E}}]
        &\le
        \E[Z \cdot \textbf{1}_{Z>3m/2}] \\
        &\le \Pr\bigg[Z \ge \frac{3m}{2}\bigg] \cdot \E\bigg[Z \mid Z\ge \frac{3m}{2}\bigg] \\
        &\le
        \exp\bigg(\frac{-m}{10}\bigg) \cdot  \E\bigg[Z \mid Z\ge \frac{3m}{2}\bigg] && (\text{Chernoff})
    \end{align*}

    To bound the conditional expectation, we use the rapid decay of the upper tail of the binomial distribution. 
    \begin{align*}
        \E\bigg[Z \mid Z\ge \frac{3m}{2}\bigg]
        &=
        \frac{3m}{2} + \sum_{\ell=1}^{m\sqrt{n}-3m/2} \Pr\bigg[Z \ge \frac{3m}{2} + \ell \mid Z \ge \frac{3m}{2}\bigg] \\
        &=
        \frac{3m}{2} + \sum_{\ell=1}^{m\sqrt{n}-3m/2} \frac{ \Pr[Z \ge \frac{3m}{2} + \ell]}{ \Pr[Z \ge \frac{3m}{2}]} \\
        &=
        \frac{3m}{2} + \sum_{\ell=1}^{m\sqrt{n}-3m/2} \left(\prod_{j=0}^{\ell-1}  \frac{ \Pr[Z \ge \frac{3m}{2} + j+1]}{ \Pr[Z \ge \frac{3m}{2} + j]}\right) && (\text{telescopic product})\\\\
        &\le
        \frac{3m}{2} + \sum_{\ell=1}^{m\sqrt{n}-3m/2} \bigg(\frac{2}{3}\bigg)^{\ell} && (\text{\Cref{cla:binomial_ratio}})\\
        &\le 
        \frac{3m}{2} + 2
    \end{align*}

    \begin{claim}\label{cla:binomial_ratio}
        For any natural $t \ge \frac{3m}{2}$, it holds that $\frac{\Pr[Z=t+1]}{\Pr[Z=t]} \le \frac23$.
    \end{claim}
    \begin{proof}
    Let $K = m \sqrt{n}$ and $p = 1/\sqrt{n}$, it holds that $Z \sim \textrm{Bin}(K,p)$, and $m = Kp$.
    \begin{align*}
        \frac{\Pr[Z=t+1]}{\Pr[Z=t]}
        &=
        \frac{\binom{K}{t+1}p^{t+1}\left(1-p\right)^{K-(t+1)}}{\binom{K}{t}p^{t}\left(1-p\right)^{K-t}} \\
        &=
        \frac{K-t}{t+1} \cdot \frac{p}{1-p} \\
        &\le
        \frac{K-3Kp/2}{3Kp/2} \cdot \frac{p}{1-p} && (t \ge 3m/2=3Kp/2)\\
        &=
        \frac{1/p-3/2}{3/2} \cdot \frac{p}{1-p} \\
        &=
        \frac23 \cdot \frac{1-3p/2}{1-p} 
        \le \frac23
    \end{align*}
    \end{proof}
    
    Putting everything together we get that
    \[
    \E[M \cdot \textbf{1}_{\neg\mathcal{E}}] \le \exp\bigg(\frac{-m}{10}\bigg)\cdot \bigg(\frac{3m}{2}+2\bigg) = o(1),
    \]
    which concludes the proof.
\end{proof}

\bibliographystyle{alpha}
\bibliography{refs}

\appendix
\crefalias{section}{appendix}

\section{\boldmath Lower Bound of $9/4$ for Additive and $2$-Demand $f$}
\label{sec:add_lower_bound}

In this section we present a lower bound on the competitive ratio of any policy for $2$-demand or additive $f$.

\additivelowerbound*

Consider the following instance with four additive agents. Agents~1 and~3 are deterministic, while agents~2 and~4 are active with probability $0 < p_i < 1$.
The construction is specified in \Cref{tab:54_lb}.

\begin{table}[h]
    \centering
    \begin{tabular}{|c|c|c|c|c|}
        \hline
        \textbf{Agent} & \textbf{1} & \textbf{2} & \textbf{3} & \textbf{4} \\
        \hline
        $p_i$ & 1 & 1/2 & 1 & 1/M \\
        \hline
        $(f_i,c_i)$ & $(1/2,1/M)$ & $(1,2/M)$ & $(M,M-1)$ & $(M^2,M^2-M)$ \\
        \hline
        $c_i/f_i$ & $2/M$ & $2/M$ & $1-1/M$ & $1-1/M$ \\
        \hline
        $\E_{p_i}[f_i-c_i]$ & $1/2-1/M$ & $1/2-1/M$ & $1$ & $1$ \\
        \hline
    \end{tabular}
    \caption{
An instance with four agents for which the competitive ratio of any algorithm is at least $9/4$.
The first row specifies the probability that agent~$i$ is active. The second row lists the corresponding parameters $(f_i,c_i)$. The third row gives the minimal contract required to incentivize agent~$i$, if active. The fourth row specifies the expected utility obtained by hiring agent~$i$ alone, if active. The parameter~$M$ tends to infinity.
}
    \label{tab:54_lb}
\end{table}

The following observation, which enumerates all sets of agents that should be considered by the principal, suggests that the above instance can be reduced to a $2$-demand instance with the individual values.
\begin{observation}\label{obs:lb_possible_sets}
    Only the sets $\{\{1\}, \{2\}, \{3\}, \{4\},\{1,2\}\}$, may yield positive utility to the principal.
\end{observation}

\begin{lemma}
    The prophet's expected utility in the example depicted in Table~\ref{tab:54_lb} is at least $9/4$.
\end{lemma}

\begin{proof}
    If agent $4$ is active, then the prophet obtains $g(\{4\}) = (1-\alpha_4)f_4=M$.
    If agent $2$ is active, the principal can incentivize the set $\{1,2\}$ for a utility of $(1-\alpha_1-\alpha_2)(f_1+f_2)=(1-4/M)(1/2+1)=3/2-6/M$.
    If neither is active, the best utility achievable is $g(\{3\})=(1-\alpha_3)f_3=1$.
    Thus, the prophet's expected utility is at least
    \[
    \frac{1}{M} \cdot M + \left(1-\frac{1}{M}\right)\left( \frac{1}{2} \cdot \left(\frac{3}{2}-\frac{6}{M}\right) + \frac{1}{2} \cdot 1\right) \overbrace{
    \longrightarrow}^{M \to \infty} \frac{9}{4}
    \]
\end{proof}

\begin{lemma}
In the example depicted in Table~\ref{tab:54_lb}, no online algorithm can achieve an expected utility better than $1$.
\end{lemma}

\begin{proof}
    Clearly, any optimal algorithm offers each active agent either
$\alpha_i = c_i / f_i$ or $\alpha_i = 0$.

Suppose the algorithm offers $\alpha_1 = c_1 / f_1 = 2/M$.
In this case, it must set $\alpha_3 = \alpha_4 = 0$; otherwise, the total payment would be at least~1, resulting in non-positive utility.

Recall that hiring only agent 1 yields utility 
$(1-2/M)\cdot 1/2=1/2-1/M$. 
If agent 2 is active, then hiring both agents 1 and 2 yields utility
$(1-2/M-2/M)(1+1/2) = 3/2-6/M$.
Therefore, the algorithm's expected utility in this case is at most
\[
\frac12 \left(\frac12-\frac{1}{M}\right)
+
\frac12 \left(\frac32-\frac{6}{M}\right) \overbrace{
\longrightarrow}^{M \to \infty} 1.
\]

Alternatively, if the algorithm skips agent~1 and offers $\alpha_1 = 0$, then by \Cref{obs:lb_possible_sets}, an optimal algorithm selects the agent maximizing
$\mathbb{E}_{p_i}\!\left[(1 - \alpha_i) f_i\right]$.
This expression is maximized by either agent~3 or~4 (see \Cref{tab:54_lb}), yielding an expected utility of~1.
This completes the proof.
\end{proof}

\section{\boldmath $O(1)$-CR for Submodular $f$: Decomposing the Benchmark}\label{subsec:decomposition}
In this section we show that  it is enough to design a constant competitive policy for the budgeted problem in order to be $O(1)$-competitive with respect to the online contracting problem.
Recall the offline optimum for the budgeted problem, defined in \Cref{eq:budget_half}:
\begin{equation*}\tag{\ref{eq:budget_half}}
\budgetHalfDefinition
\end{equation*}

\reductiontobudget* 

The reduction in \cref{prop:reduction_to_budget} relies on the following decomposition of the benchmark.
\begin{lemma}\label{lem:decomposition}
For any submodular function $f$, and any set of active agents $A$, the maximum between $g(\OPTbudget)$ and $g(i^\star)$ is a 3-approximation to $g(\OPT)$, where 
$\OPTbudget$ is as defined in \Cref{eq:budget_half}, $i^\star = \arg\max_{i \in A} g(\{i\})^+$, and $\OPT$ is the optimal solution to the offline contracting problem induced by $A$.
\end{lemma}
\begin{proof}
Fix a realized set $A \subseteq N$ of active agents, and let $A' = \{i \in A \mid \frac{c_i}{f_i} \le \frac{1}{2}\} = A \cap L$.
Observe that,
$
g(\OPTbudget) 
\ge 
f(\OPTbudget)/2
$.
Moreover, for any set $U \subseteq A'$ and $\sum_{i\in U} \frac{c_i}{\marg{i}{U}} \le 1/2$, it holds that $f(\OPTbudget) \ge f(U)$, thus 
\begin{equation}\label{eq:Kstar_to_U}
 g(\OPTbudget) 
\ge \frac{1}{2} f(\OPTbudget)
\ge \frac{1}{2} f(U)
\ge \frac{1}{2} g(U)   
\end{equation}
Let $M = \max\{g(\OPTbudget),\max_{i \in A}g(\{i\})^+\}$. 
We show that $g(\OPT)\le 3M$.

If $|\OPT \setminus A'| \ge 2$, 
then by submodularity $\sum_{i \in \OPT} c_i/\marg{i}{\OPT} \ge \sum_{i \in \OPT \setminus A'} c_i/f_i \ge 1$. Thus, $g(\OPT) \le 0$, and the claim trivially holds.

If $|\OPT \setminus A'| = 1$, then $\sum_{i \in \OPT \setminus A'} \frac{c_i}{\marg{i}{S}} \le 1/2$, or $g(\OPT) \le 0$. 
Together with the subadditivity of $g$ and \Cref{eq:Kstar_to_U} we get,
\[
g(\OPT) 
\le 
g(\OPT \setminus A') + g(\OPT \cap A') 
\le 
2\cdot g(\OPTbudget)+\max_{i \in A} g(\{i\}) \le 3M.
\]
If $\OPT \setminus A' = \emptyset$, and $\sum_{i \in \OPT} c_i/\marg{i}{\OPT} \le 1/2$, then 
by \Cref{eq:Kstar_to_U},
\[
M \ge g(\OPTbudget) \ge \frac12 g(\OPT)
\]
If $\OPT \setminus A' = \emptyset$, 
and $\sum_{i \in \OPT} c_i/\marg{i}{\OPT} > 1/2$, then observe that
$g(\OPT) < (1/2) \cdot f(\OPT)$.
Moreover, there exists a partition of $\OPT = \OPT_1 \sqcup \OPT_2 \sqcup \OPT_3$ such that 
$\sum_{i \in \OPT_j} c_i/\marg{i}{\OPT_1} \le 1/2$, for any $j \in [3]$:
Consider elements in an arbitrary order and add them to one set, $\OPT_1$, while $\sum_{i \in \OPT_1} \frac{c_i}{\marg{i}{\OPT}} < 1/2$. Do the same with the remaining elements and $\OPT_2$, and finally set $\OPT_3 = \OPT \setminus (\OPT_1 \cup \OPT_2)$.
Assume without loss of generality that $f(\OPT_1) \ge f(\OPT)/3$.
Thus,
$$
M \ge g(\OPTbudget) 
\ge 
\frac{1}{2}f(\OPT_1) \ge \frac{1}{6}f(\OPT) > \frac{1}{3}g(\OPT),
$$
where the second inequality follows from \Cref{eq:Kstar_to_U}.
This concludes the proof.
\end{proof}

\begin{proof}[Proof of \Cref{prop:reduction_to_budget}] 
    Recall that there exists $2$-competitive for the online problem of hiring a single agent (i.e., unit demand $f$), see \Cref{obs:unit-demand-2}.
    Consider an algorithm that with probability $2/(\gamma+2)$ runs 
    a $2$-competitive unit-demand prophet algorithm, and with probability $1-2/(\gamma+2)$ runs a $\gamma$-competitive algorithm for the budgeted problem. It's expected utility is

    \[
    \frac{2}{(\gamma+2)} \cdot \frac12 \E\left[\max_{i \in A} g(i) \right] + 
    \frac{\gamma}{(\gamma+2)}
    \cdot \frac{1}{\gamma} \E\left[f(\OPTbudget)\right] 
    \ge 
    \frac{1}{\gamma+2} \max\left\{\E\left[\max_{i \in A} g(i) \right], \E\left[g(\OPTbudget)\right]\right\} 
    \ge  \frac{\E[g(OPT)]}{3(\gamma+2)},
    \]
    where the last inequality follows from \Cref{lem:decomposition}.
\end{proof}

\section{Dominant Strategies Under Safe Policies}
\label{sec:safe_equilibrium}
In this section we complete the details regarding the solution concepts induced by safe policies under different classes of $f$.

\safeimpliesds*

\begin{proof}[Proof of \Cref{obs:safe_implies_ds}]
Fix an arbitrary realization of the active agents and the internal randomness of the policy. 
Let $\contract$ be the outcome if a \safe policy and let $S^{hire}$ be the resulting set of hired agents.

Fix a hired agent $i\in S^{hire}$, and consider an arbitrary set
$T\subseteq S^{hire}\setminus \{i\}$
of other hired agents who exert effort. 
It is a (weakly) best response for agent $i$ to exerts effort if $\alpha_i \ge \frac{c_i}{\marg{i}{T}}$.
By submodularity and since $\alpha_i$ is the output of a safe policy we have
$
\alpha_i \ge \frac{c_i}{\marg{i}{S^{hire}}} \ge \frac{c_i}{\marg{i}{T}}
$.
Thus, for every possible set $T\subseteq S^{hire}\setminus \{i\}$ of other hired agents who exert effort, agent $i$ weakly prefers exerting effort to shirking.
Since the choice of $T$ was arbitrary, exerting effort is a weakly dominant strategy for agent $i$. 
\end{proof}

\begin{proposition}\label{prop:safe_XOS_not_DS}
There exists an online-contract instance with strictly positive costs and an XOS reward function for which a safe policy does not induce exerting effort as a dominant strategy for every hired agent.
\end{proposition}

\begin{proof}
Consider an instance with three agents, denoted by $N=\{1,2,3\}$, and suppose that all agents are active with probability $1$, i.e., $p_i=1$ for every $i\in N$.
Define $f:2^N\to[0,1]$ to be $f(S)=\max\{a(S),b(S)\}$, for the two additive functions
$$
a(S)=\frac{4}{9}\cdot \mathbf{1}_{\{2\in S\}}
 \qquad\text{and } \qquad 
b(S)=
\frac{4}{9}\cdot \mathbf{1}_{\{1\in S\}}
+\frac{1}{9}\cdot \mathbf{1}_{\{2\in S\}}
+\frac{4}{9}\cdot \mathbf{1}_{\{3\in S\}}.
$$
Note that $f$ is normalized, monotone, and XOS.
Let the costs be
$$
c_1=\frac{2}{27},
\qquad
c_2=\frac{1}{72},
\qquad
c_3=\frac{1}{18}.
$$
Consider the deterministic policy that hires all three agents and offers the contract vector
$$
\alpha_1=\frac{1}{6},
\qquad
\alpha_2=\alpha_3=\frac{1}{8}.
$$
We first show that the policy is safe. Since the policy hires all agents, we have
$ S^{hire}=N$.
We compute the marginal contribution of each agent with respect to $S^{hire}$. Observe that,
$$
\marg{1}{\{2,3\}}
=
\frac{4}{9},
\qquad 
\marg{2}{\{1,3\}}
=
\frac{1}{9},
\qquad
\marg{3}{\{1,2\}}
=
\frac{4}{9}.
$$
Thus,
$$
\frac{c_1}{\marg{1}{\{2,3\}}}= \frac16=\alpha_1,
\qquad 
\frac{c_2}{\marg{2}{\{1,3\}}}= \frac18=\alpha_2,
\qquad 
\frac{c_3}{\marg{3}{\{1,2\}}}= \frac18=\alpha_3,
$$
Hence, each hired agent is incentivized to exert effort with respect to the full hired set $S^{hire}=N$, and the policy is safe. In particular, if agents $2$ and $3$ both exert effort, the best response of agent $1$ is to exert effort as well.

We now show that exerting effort is not a best response for agent $1$ when agent $2$ exert effort and agent $3$ does not. In particular, working is not a best response for agent $1$.
First, note that
$
\marg{1}{\{2\}} = \frac{1}{9}
$.
Therefore, we have that $\frac{c_1}{\marg{1}{\{2\}}} = \frac{2}{3} > \alpha_1$, which implies that shirking is a strict best-response for agent $1$.
This concludes the proof.
\end{proof}

\section{Tight 2-Competitive Ratio for Unit Demand $f$ }\label{sec:unit-demand}

In this section we show that the problem of online contracting with a single agent is essentially the classic prophet inequality setting of \cite{krengel1977semiamarts,samuel1984comparison}, and thus admits a (tight) competitive ratio of $2$.

Equivalently, we consider the online contracting problem with a unit-demand $f$, that is, $f(S) = \max_{i \in S} f_i$.
The optimal solution, in any realization, is to take at most one agent. Since the optimal contract is of the form $c_i/f_i$ for non-zero $f_i$, we have that the utility from offering the optimal contract to agent $i$ is $(1-c_i/f_i)f_i = f_i - c_i$.
Thus, from the prophet's perspective, the problem reduces to the classic single-item prophet inequality with non-negative values $v_i = (f_i-c_i)^+$, where $(x)^+ = \max\{x,0\}$ and each item $i$ has probability of $p_i$ of being available.
Thus, a competitive ratio of $2$ is achievable by applying any standard prophet inequality algorithm for this setting.

\begin{observation}\label{obs:unit-demand-2}
    When $f$ is unit-demand, a competitive ratio of 2 is achievable.
\end{observation}

Notably, the lower bound on the competitive ratio for this problem is not immediate. Unlike the single-item setting, the principal may select multiple items (i.e., incentivize more than one agent), while her reward is determined by the highest $f_i$ among the incentivized agents; see Observation~\ref{obs:not-safe}.
In principle, this additional flexibility could allow the principal to achieve a better competitive ratio than in the standard gambler setting. However, the example below, which closely resembles the lower bound construction for the classic prophet inequality, shows that in the worst case, this added power does not improve upon the competitive ratio of~$2$.

\begin{claim}
    For any $\eps > 0$, there exists an instance with two agents such that any algorithm achieves a competitive ratio of at least $2-\eps$.
\end{claim}
\begin{proof}
    Consider a setting with two agents and a unit-demand $f$, with the following in order:
    Agent 1 is deterministic with $p_1=1$, with $f_1=\eps$ and $c_1 = \frac{\eps}{2}$. 
    Agent 2 has activation probability $p_2 = \eps$,  with $f_2=1$, and $c_2=1/2$.
    
    Clearly, the prophet's revenue is
    $$
    (1-\eps)\frac{\eps}{2} + \eps \frac{1}{2} = \frac{\eps}{2}(2-\eps).
    $$
    Observe that any optimal algorithm offers each of the agent the contract $\alpha_i=0$ or $\alpha_i = \frac{c_i}{f_i} = \frac{1}{2}$.
    Thus, if $\alpha_1 = \frac{1}{2}$, even if agent 2 is active, the principal would offer $\alpha_2=0$, as otherwise her profit is
    $(1-\alpha_1-\alpha_2)\frac{1}{2} = 0$.
    Thus, when $\alpha_1 = \frac{1}{2}$, the principal's expected revenue is $\frac{\eps}{2}$.
    If $\alpha_1=0$, her maximal expected reward is given by $(1-1/2)\cdot 1 \cdot \eps = \eps/2$, and the claim follows.
\end{proof}

\section{Zero Costs Incur a Multiplicative Loss of at most 2}
\label{sec:positive_costs}
\begin{proposition}\label{prop:positive_costs}
Suppose that for every online-contract instance with monotone submodular reward function and strictly positive costs, i.e., $c_i>0$ for all agents $i$, there exists a $\rho$-competitive online policy. Then, for every online-contract instance with monotone submodular reward function and nonnegative costs, i.e., $c_i\ge 0$ for all agents $i$, there exists a $2\rho$-competitive online policy.
\end{proposition}

\begin{proof}
Fix an instance
$\mathcal I=\langle f,\{c_i\}_{i\in[n]},\{p_i\}_{i\in[n]}\rangle$
with $c_i\ge 0$ for every agent $i$. Let
$P=\{i\in[n]\mid c_i>0\}$ and 
$Z=[n]\setminus P$
be the sets of positive-cost and zero-cost agents, respectively.

Consider the restricted instance on the positive-cost agents,
$\mathcal{I}^P=\langle f^P,\{c_i\}_{i\in P},\{p_i\}_{i\in P}\rangle$,
where $f^P(S)=f(S)$ for every $S\subseteq P$, and where the arrival order is the order inherited from the original instance. Since $f$ is monotone submodular, so is $f^P$. By assumption, there exists a $\rho$-competitive online policy $\pi^P$ for $\mathcal{I}^P$.

We define a policy $\pi$ for the original instance $\mathcal{I}$ as follows. Before the process begins, $\pi$ flips a fair coin.

If the coin lands heads, $\pi$ runs the zero policy: it offers $\alpha_i=0$ to every active agent $i$.

If the coin lands tails, $\pi$ simulates $\pi^P$ on the positive-cost agents only. Namely, whenever an active agent $i\in P$ arrives, $\pi$ feeds this arrival to the simulated policy $\pi^P$ and offers the contract chosen by $\pi^P$. Whenever an active agent $i\in Z$ arrives, $\pi$ offers $\alpha_i=0$ and does not feed this arrival to the simulation.

We first lower-bound the expected utility of $\pi$. Let $A$ denote the random set of active agents, and write
$A^P=A\cap P$ and $A^Z=A\cap Z $.

Consider first the branch in which $\pi$ runs the zero policy. In this branch all active agents receive contract $0$. No positive-cost agent is willing to exert effort, while every zero-cost agent is indifferent between exerting effort and shirking. Therefore, $A^Z$ is an equilibrium. By the principal-favoring tie-breaking rule, and since $f$ is monotone, the equilibrium selected by the agents gives the principal utility at least $f(A^Z)$.
Thus, the expected utility of this branch is at least
$\mathbb{E}_A[f(A^Z)]$.

Next consider the branch in which $\pi$ simulates $\pi^P$. Fix a realization $A$ of the active set and fix the internal randomness of $\pi^P$. Let $\alpha^P$ be the contracts produced by $\pi^P$ for the active positive-cost agents $A^P$, and let $S^P\subseteq A^P$ be the principal-preferred equilibrium in the restricted instance $\mathcal{I}^P$. In the original instance, extend the contract vector by setting $\alpha_i=0$ for every active zero-cost agent $i\in A^Z$.

We claim that $S^P$ is also an equilibrium in the original instance. Indeed, for every positive-cost agent $i\in A^P$, the relevant marginal contributions with respect to $S^P$ are exactly the same as in the restricted instance, since no zero-cost agent is included in $S^P$. Therefore, the equilibrium inequalities for agents in $A^P$ remain unchanged. For every zero-cost agent $i\in A^Z$, we have $c_i=0$ and $\alpha_i=0$, so the agent is indifferent between exerting effort and shirking. Hence shirking is a weak best response. Therefore, $S^P$ is an equilibrium in the full instance.

By the principal-favoring tie-breaking rule, the equilibrium selected in the full instance gives the principal utility at least the utility obtained from $S^P$. Therefore, the expected utility of the simulated branch is at least the expected utility of $\pi^P$ on the restricted instance $\mathcal{I}^P$. Since $\pi^P$ is $\rho$-competitive,
$$ALG(\pi^P,\mathcal{I}^P)\ge \frac{1}{\rho}OPT(\mathcal{I}^P).$$
Writing
$
OPT^P=
\E_A\left[\max_{S\subseteq A^P} g(S)\right]$,
we have
$$ALG(\pi^P,\mathcal{I}^P)\ge \frac{OPT^P}{\rho}.$$
Combining the two branches, we obtain
$$ALG(\pi,\mathcal{I})
\ge
\frac12\cdot \frac{OPT^P}{\rho}
+
\frac12\cdot \E_A[f(A^Z)].$$
Since $\rho\ge 1$, this implies
$$ALG(\pi,\mathcal{I}) \ge
\frac{1}{2\rho}
\left(
OPT^P+\E_A[f(A^Z)]
\right).$$

It remains to compare this quantity with the offline optimum of the original instance. For a fixed realization $A$, let
$$S^\star(A)\in \arg\max_{S\subseteq A} g(S)$$
be an optimal offline working set. Decompose
$$S^\star(A)=S_P^\star\cup S_Z^\star,$$
where $S_P^\star=S^\star(A)\cap P$ and $S_Z^\star=S^\star(A)\cap Z$.

By subadditivity of $g$, which follows from monotone submodularity of $f$, we have
$$g(S^\star(A))
\le
g(S_P^\star)+g(S_Z^\star).$$

Since $S_P^\star\subseteq A^P$,
$$g(S_P^\star)
\le
\max_{S\subseteq A^P} g(S).$$
Moreover, since every agent in $S_Z^\star$ has zero cost,
$$g(S_Z^\star)=f(S_Z^\star)\le f(A^Z),$$
where the inequality follows from monotonicity of $f$. Therefore, for every realization $A$,
$$\max_{S\subseteq A} g(S)
=
g(S^\star(A))
\le
\max_{S\subseteq A^P} g(S)+f(A^Z).$$
Taking expectation over the random active set $A$, we get
$$OPT(\mathcal I) \le OPT^P+\E_A[f(A^Z)].$$

Combining this with the lower bound on $ALG(\pi,\mathcal I)$ gives
$$ALG(\pi,\mathcal I)
\ge
\frac{1}{2\rho}OPT(\mathcal I).$$
Thus $\pi$ is $2\rho$-competitive for the original instance with nonnegative costs.
\end{proof}

\section{Computing the Expected Maximum Load}\label{sec:max-load-proof}
\begin{proposition}\label{prop:max-load-n}
In the notation of \Cref{sec:xos} ($\sqrt n$ columns $C_1,\dots,C_{\sqrt n}$, with $X_j = |A\cap C_j|$ i.i.d.\ $\sim \mathrm{Bin}(\sqrt n, 1/\sqrt n)$),
\[
\E\Big[\max_{j\in[\sqrt n]} X_j\Big] = \Theta \left(\frac{\log n}{\log\log n}\right).
\]
\end{proposition}
\begin{proof}
Let $j^\star \in \argmax_{j\in[\sqrt n]} X_j$

We first upper bound of $\E[X_{j^\star}]$.
For any $j \in \sqrt{n}$ we have that $X_j \sim Bin(1/\sqrt{n}, \sqrt{n})$, thus, for any integer $U \in [\sqrt{n}]$, we have by Stirling's approximation
\[ 
\Pr[X_j \ge U] \le \binom{\sqrt{n}}{U}\cdot \left(\frac{1}{\sqrt{n}}\right)^U\le \frac{e^U}{U^U}
\]
and by the union bound
\[
\Pr[X_{j^\star} \ge U] \le \sqrt{n} \cdot \frac{e^U}{U^U} = \exp(\ln \sqrt{n} + U - U \ln U).
\]
For $U = \frac{4\ln \sqrt{n} }{\ln \ln \sqrt{n}}$ and large enough $n$, we get $\Pr[X_{j^\star} \ge U] \le \exp(-3\ln \sqrt{n})=\frac{1}{n^{3/2}}$.
Thus, 
\[
\E[X_{j^\star}] \le U + \sqrt{n} \cdot \Pr[X_{j^\star} \ge U] \le U + \frac{1}{n} = O\left( \frac{\log n}{\log \log n}\right)
\]

The lower bound of $\E[X_{j^\star}]$ is identical to the argument in \Cref{prop:xos-prophet-lower-bound}:
Let $Y_j = 1$ if $X_j \ge L$, and let $L = \frac{\ln \sqrt{n}}{3 \ln \ln \sqrt{n}}$.
\[
\Pr[Y_j = 1] \ge \Pr[X_j = L] = \binom{\sqrt{n}}{L} \left( \frac{1}{\sqrt{n}} \right)^L \left( 1 - \frac{1}{\sqrt{n}} \right)^{n-L} \geq \left(\frac{\sqrt{n}}{L}\right)^L \frac{1}{\sqrt{n}^L} \frac{1}{e} = \frac{1}{L^L e} = \frac{1}{e n^{1/6}}.
\]
So $\mu := \sum_{j \in [\sqrt{n}]} \Pr[Y_j = 1] \geq \frac{n^{1/3}}{e}$. As the $Y_j$ are independent and their variance is bounded by their expectation, Chebyshev's inequality gives us
\[
\Pr\left[\sum_{j \in [\sqrt{n}]} Y_j = 0\right] \leq \Pr\left[\left\lvert\sum_{j \in [\sqrt{n}]} Y_j - \mu \right\rvert \geq \mu\right] \leq \frac{\mathrm{Var}(\sum_{j \in [\sqrt{n}]} Y_j)}{\mu^2} \leq \frac{1}{\mu} \leq \frac{e}{n^{1/3}}.
\]

Consequently,
\[
\E[\max_j \min\{K, X_j\}] \geq \left(1 - \frac{e}{n^{1/3}} \right) L = \Omega\left(\frac{\log n}{\log \log n}\right).
\]
This concludes the proof.
\end{proof}

\section{Proof of the Expected XOS Marginal Lemma}\label{sec:expected-xos-margs-proof}

\begin{proof}[Proof of \Cref{lem:expected_XOS_margs}]
	For every realization $A\subseteq[n]$, applying \Cref{lem:XOS_margs} with $S=T=A$ gives 
    $$\sum_{i\in A}\big(f(A)-f(A\setminus\{i\})\big) = \sum_{i\in A} \marg{i}{A} \le f(A)$$
	Taking expectation over the random set $A$, we obtain
	\begin{align*}
		\E_A[f(A)]
		&\ge
		\E_A \left[
			\sum_{i\in A}\big(f(A)-f(A\setminus\{i\})\big)
		\right] \\
		&=
		\sum_{i=1}^n
		\E_A \left[
			\mathbf{1}[i\in A]\cdot
			\big(f(A)-f(A\setminus\{i\})\big)
		\right].
	\end{align*}
	Now fix an agent $i$. Since the event $i\in A$ is independent of the realization of $A_{-i}$, and since whenever $i\in A$ we have $A=A_{-i}\cup\{i\}$, it follows that
	\[
		\E_A\left[
			\mathbf{1}[i\in A]\cdot
			\big(f(A)-f(A\setminus\{i\})\big)
		\right]
		=
		p_i\cdot
		\E_{A_{-i}}\big[f(A_{-i}\cup\{i\})-f(A_{-i})\big].
	\]
	Substituting this equality into the previous display yields
	\[
		\E_A[f(A)]
		\ge
		\sum_{i=1}^n p_i\cdot
		\E_{A_{-i}}\big[f(A_{-i}\cup\{i\})-f(A_{-i})\big],
	\]
	as claimed.
\end{proof}

\section{\boldmath Proof that \texorpdfstring{$h$}{h} is XOS}\label{sec:h_is_XOS_proof}

\hisXOS*

\begin{proof}[Proof of \Cref{clm:h_is_XOS}]
    Let us rename the agents of $X$ so that $X = \{1,\dots,|X|\}$, and denote $X^i = \{1,\dots,i-1\}$.

    For any $Z \subseteq X$, define the linear function $\ell^Z \in \reals^{|X|}_{\ge0}$ to be $\ell^Z = \boldsymbol{0}$ if $\sum_{i \in Z}\frac{c_i}{\marg{i}{S}} > 1/2$, and $\ell^Z_i = f(i \mid Z \cap X^i) \cdot I[i \in Z]$, otherwise.
    Fix any $X' \subseteq X$.
    If $S \subseteq X'$ is the set that maximizes \Cref{eq:def_h}, then
    $$h(X')=f(S)=\sum_{i \in S} f(i \mid S \cap X^i) = \sum_{i \in X'} \ell^S_i = \ell^S(X').$$
    Additionally, if there exists a linear function $\ell^Z$ such that $\ell^Z(\smallAgents') > \ell^S(\smallAgents')$,
    it holds that (i) payment is at most $1/2$: $\sum_{i \in X' \cap Z} \frac{c_i}{\marg{i}{Z}} \le \sum_{i \in Z} \frac{c_i}{\marg{i}{Z}} \le 1/2$ (or $\ell^Z(X')=0$), and
    (ii) reward is greater than $f(S)$:
    \begin{align*}
    f(S)
    &=
    \ell^S(X')
    < \ell^Z(X')
    =
    \sum_{i \in X'} \ell^Z_i \cdot I[i \in Z]
    =
    \sum_{i \in Z \cap X'} f(i \mid Z \cap X^i) \\
    &\le
    \sum_{i \in Z \cap X'} f(i \mid Z \cap X' \cap X^i)
    =
    f(Z \cap \smallAgents'),
    \end{align*}
    where the last inequality follows from submodularity of $f$. The existence of the set $X' \cap Z$ contradicts the optimality of $S$, and the claim follows.
\end{proof}

\section{\boldmath Subadditivity of $g$}\label{sec:gSA}

\gsubadditive*

\begin{proof}[Proof of \Cref{lem:gSA}]
    Let $S,T \subseteq [n]$ and assume $f$ is submodular. It holds that
    \begin{align*}
    g(S\cup T) 
    &=
    \left( 1-\sum_{i \in S\cup T} c_i/\marg{i}{S\cup T} \right)f(S\cup T) \\
    &=
    \left( 1-\sum_{i \in S\cup T} c_i/\marg{i}{S\cup T} \right)(f(T)+f(S \setminus T \mid T)) \\
    &\le
    \left( 1-\sum_{i \in S\cup T} c_i/\marg{i}{S\cup T} \right)(f(T)+f(S \setminus T )) && (\text{submodularity of } f) \\
    &\le
    \left( 1-\sum_{i \in S\cup T} c_i/\marg{i}{S\cup T}\right)(f(T)+f(S)) && (\text{monotonicity of } f) \\
    &\le
    \left( 1-\sum_{i \in T} c_i/\marg{i}{S\cup T}\right)f(T)+ \left( 1-\sum_{i \in S} c_i/\marg{i}{S\cup T} \right)f(S) && (c_i/\marg{i}{S\cup T} \ge 0) \\
    &\le
    \left( 1-\sum_{i \in T} c_i/\marg{i}{T}\right)f(T)+ \left( 1-\sum_{i \in S} c_i/\marg{i}{S} \right)f(S) && (\text{submodularity of }f)\\
    &=
    g(T) + g(S),
    \end{align*}
    which completes the proof.
\end{proof}

\section{\boldmath Constant-Competitive Ratio for WMRF $f$ via Safe Policies}\label{sec:WMRF}

In this section we show that when $f$ is a weighted matroid rank function there exists a constant-competitive safe policy.

\begin{theorem}\label{thm:wmrf}
There exists a safe policy that achieves $58.2$-competitive ratio when $f$ is a weighted matroid rank function.
\end{theorem}

When $f$ is additive, we prove that a safe policy achieves an improved competitive ratio.

\begin{theorem}\label{thm:additive}
There exists a safe policy that achieves $20$-competitive ratio when $f$ is additive.
\end{theorem}

Similar to the submodular case (\Cref{sec:positive_SM}), we first decompose the problem into finding the best single agent, and a budgeted problem.
The extra structure of WMRF allows us to reduce the budgeted problem to a knapsack problem with \emph{fixed} sizes, which does not depend on the set of elements already chosen.
As such, known OCRS for the additive objective under knapsack and matroid constraints can be combined to solve the latter sub-problem.

\newcommand{\optks}{S^{\textsf{KS}}}
\newcommand{\optsin}{S^{\textsf{SIN}}}

First we formally define the class of weighted matroid rank functions.
\begin{definition}[Weighted Matroid Rank Function (WMRF)]
A function $f:2^{[n]} \to \reals_+$ is a weighted matroid rank function if there exists a matroid on the ground set $[n]$, $\mathcal{M} = ([n],\mathcal{I})$ and a collection of non-negative weights $f_1, \dots,f_n$ such that for any $S \subseteq [n]$, $f(S)$ is the weight of the heaviest independent set in $S$. Namely,
\[
f(S) = \max_{S' \subseteq S, S'\in \mathcal{I}} \sum_{i \in S'} f_i.
\]
\end{definition}

We first observe that when $f$ is a WMRF, the optimal contract when incentivize an independent set of agents is to offer each agent $i$ a contract $\alpha_i = c_i/f_i$. This is because the marginal contribution of each agent $i$ in an independent set $S$ is exactly $f_i$.

Similar to the constant-competitive ratio for submodular functions, we reduce the problem to a budgeted variant of the problem. 
We leverage the structure of the optimal contract for an independent set of agents to show that the budgeted variant is closely related to the classic knapsack problem with values $f_i$ and sizes $s_i = c_i/f_i$.

\begin{definition}[(offline) Knapsack with weighted matroid rank function $f$]
Given a set of $n$ items with sizes $\{s_i\}_{i \in [n]}$, a WMRF
$f \colon 2^{[n]} \to \mathbb{R}$ with singleton values
$\{f_i = f(\{i\})\}_{i \in [n]}$, and a budget $B > 0$, the knapsack problem
asks for a set $S \subseteq [n]$ that maximizes $f(S)$ subject to
$\sum_{i \in S} s_i \le B$.
We say that an instance contains only \emph{small items} if
$s_i < B/2$ for every $i \in [n]$.
\end{definition}

The online variant of the problem has the same flavor of the online contracting problem.
\begin{definition}[(online) Knapsack with weighted matroid rank function $f$]
Fix a knapsack capacity $B > 0$. 
In the online variant of the problem, the value and size of each item, $(f_i, s_i)$, are fixed and known.
Items arrive in an order known to the algorithm, and each item $i$ is active with probability $p_i$, independently of other items.
Upon arrival, the algorithm must decide immediately and
irrevocably whether to accept item $i$.
We say that an instance contains only \emph{small items} if for every $i \in [n]$ we have $s_i < B/2$.
\end{definition}

Our main observation is that it is sufficient to find a competitive algorithm for the online knapsack problem with small items and budget $B=2/3$.
\begin{proposition}\label{cor:ks_to_online_contract}
    If a competitive ratio $\beta \ge 1$ is achievable for the online knapsack problem with small items and $B=2/3$, then $4(\beta+2)$-CR is achievable for the online contracting problem.
\end{proposition}

The following lemma limits the number of ``large'' agents in the optimal solution to the online contracting problem.
\begin{lemma}\label{lem:two_large_agents}
    Let $T$ be an optimal solution to the offline contracting problem for a fixed realization, and let $A' = \{i \in [n] \mid (c_i/f_i)>1/3\}$.
    If $|T \cap A'| \ge 2$, then $f(T \setminus A') = 0$.
\end{lemma}

\begin{proof}
First observe that, by the optimality of $T$, we must have $|T \cap A'| < 3$.
Indeed, if $|T \cap A'| \ge 3$, then $\sum_{i \in T} c_i/f_i > 1$, which implies $g(T) < 0 = g(\emptyset)$, contradicting optimality.

We therefore assume $|T \cap A'| = 2$, and write $T \cap A' = \{j,k\}$.
Without loss of generality, assume that $T$ is independent and that
$f_i > 0$ for all $i \in T$.

By optimality of $T$, the marginal contribution of $j$ to $g$ is non-negative. Hence,
\begin{align*}
    0 
    &\le g(T) - g(T \setminus \{j\}) \\
    &= \left(1-\sum_{i \in T} \frac{c_i}{f_i}\right)f(T) - \left(1-\sum_{i \in {T\setminus \{j\}}} \frac{c_i}{f_i}\right)(f(T)-f_j) \\
    &= -\frac{c_j}{f_j}f(T)+f_j \left(1-\sum_{i \in {T\setminus \{j\}}} \frac{c_i}{f_i}\right)\\
    &< -\frac{1}{3}f(T)+f_j \frac23 \\
    &=\frac13 f_j -\frac{1}{3}f(T\setminus \{j\}),
\end{align*}
where the strict inequality follows since $c_j/f_j > 1/3$ and $\sum_{i \in T \setminus \{j\}} c_i/f_i \ge c_k/f_k > 1/3$.

It follows that $f_j \ge f(T \setminus \{j\})$.
By symmetry, the same argument applied to $k$ yields
$f_k \ge f(T \setminus \{k\})$.
Combining the two inequalities, we obtain $f_k \ge f(T \setminus \{k\})
\ge f_j
\ge f(T \setminus \{j\})
\ge f_k$, which implies $f(T \setminus \{k\}) = f_j$.

Since $j \in T \setminus \{k\}$, we have that for every $i \in T \setminus \{j,k\} = T \setminus A'$,
$f_T(i) = f(T) - f(T\setminus\{i\}) \le f(T\setminus \{k\}) - f(T\setminus\{i,k\})=0$. 
By optimality of $T$, this implies that the payment to $i$ is zero, and hence
$c_i = 0$. By assumption, this further implies that $f_i = 0$.
We conclude that $f(T \setminus A') \le \sum_{i \in T \setminus A'} f_i =0$, as claimed.
\qedhere
\end{proof}

\begin{lemma}\label{lem:add_f_approx}
    For any realizations of active agents, the maximum between the best single-agent, $\max_j (f_j-c_j)^+$, and the optimal solution to the offline knapsack problem with $B=2/3$, values $f_i$, sizes $s_i = c_i/f_i$, and where we only consider agents with $s_i \le 1/3$, achieves a $4$-approximation to the optimal contract.
\end{lemma}
\begin{proof}
Fixed a realized set of active agents, and let $A' = \{i \in [n] \mid c_i/f_i \ge 1/3\}$. Note that any agent not in $A'$ will be a small item for $B=2/3$.

Let $\optks$ be the optimal solution to the knapsack problem with $B = 2/3$ and only the agents satisfying $s_i = c_i/f_i <1/3$, i.e., an instance with small items. 

Observe that,
$g(\optks) = (1-\sum_{i \in \optks}c_i/f_i)f(\optks)\ge f(\optks)/3$.
Also, for any set $U \subseteq [n]$ with  $c_i/f_i<1/3$ for any $i \in U$ and $\sum_{i\in U} \frac{c_i}{f_i} \le 2/3$, it holds that $f(\optks) \ge f(U)$, thus 
$g(\optks) 
\ge \frac{1}{3} f(\optks)
\ge \frac{1}{3} f(U)
\ge \frac{1}{3} g(U)$.

Denote by $M$ the maximum between $g(\optks)$ and the maximal utility from hiring a single agent, i.e., $M = \max\{g(\optks),\max_{j}(f_j-c_j)^+\}$, and let $T$ be the optimal solution to the online contracting problem.
Assume without loss of generality that $T$ is an independent set and that $f_i>0$ for any $i \in T$.

If $T$ contains 3 or more agents with $c_i/f_i \ge 1/3$, then $g(T) \le 0$, and the claim trivially holds.

If $T \cap A' = \emptyset$, and $\sum_{i \in T} c_i/f_i \le 2/3$, then $T$ is a valid solution to the knapsack problem with small items. Thus,
\[
M \ge g(\optks) \ge \frac13 g(T)
\]

If $T \cap A' = \emptyset$, 
and $\sum_{i \in T} c_i/f_i > 2/3$, then observe that
$g(T) < (1/3) \cdot f(T)$.
Moreover, there exists a partition of $T = T_1 \sqcup T_2$ such that 
$\sum_{i \in T_1} c_i/f_i \le 2/3$, and $\sum_{i \in T_2} c_i/f_i \le 2/3$ (simply add elements in an arbitrary order into one set until hitting or exceeding a total payment of $1/3$. Add the remaining elements into the other set). Assume without loss of generality that $f(T_1) \ge f(T)/2$.
Thus,
$$
M \ge g(\optks) \ge \frac{1}{3}f(T_1) \ge \frac{1}{6}f(T) \ge \frac{1}{2}g(T).
$$

If $|T \cap A'| = 1$, let $\{j\} = T \cap A'$, and observe that $T\setminus \{j\}$ is such that $\sum_{i \in T\setminus \{j\}} c_i/f_i \le 2/3$ and all agents have $c_i/f_i < 1/3$. Thus, by subadditivity of $g$,
\[
g(T) \le g(T \setminus A') + g(T \cap A') \le 3g(\optks)+\max_{j \in [n]} (f_j-c_j) \le 4M
\]

If $|T \cap A'| = 2$, then by \Cref{lem:two_large_agents}, we have $f(T \setminus A')=0$.
Thus, by the subadditivity of $g$,
\[
g(T) \le g(T\setminus A')+\sum_{j \in T \cap A'} g(\{j\}) \le f(T\setminus A') + 2 \cdot \max_{j \in [n]} (f_j-c_j) \le 2M.
\]
This concludes the proof.
\end{proof}

\newcommand{\valks}{V^{\textsf{KS}}}
\newcommand{\valsin}{V^{\textsf{SIN}}}

We are now ready to prove \Cref{cor:ks_to_online_contract}.
\begin{proof}[Proof of \Cref{cor:ks_to_online_contract}]
For a given realization, let $OPT$ denote the value of the optimal principal utility, $\valks$ denotes the value the optimal knapsack solution with budget $2/3$ and small items, and $\valsin$ the value of the optimal singleton.
By \Cref{lem:add_f_approx}, for any realization $OPT \le 4 \cdot \max\{\valks,\valsin\} \le 4 (\valks+\valsin)$.
Taking expectation over both sides we get
$\E[OPT] \le 4 (\E[\valks] + \E[\valsin])$.

By \Cref{obs:unit-demand-2}, we can get a 2-CR to $\E[\valsin]$ and by assumption there is a $\beta$-CR for the knapsack problem with small agents.
The algorithm which takes the maximum from the two above guarantees a $(4(\beta+2))$-CR:
\begin{align*}
    \E[ALG]
    &=
    \max \{(1/\beta)\E[\valks], (1/2)\valsin\} \\
    &\ge 
    (\beta/(\beta+2))\cdot(1/\beta)\cdot \E[\valks] +  (2/(\beta+2))\cdot(1/2)\cdot \E[\valsin] \\
    &= (1/(\beta+2))(\E[\valks] + \E[\valsin]) \\
    &\ge 
    \E[OPT]/(4(\beta+2)),
\end{align*}
where the last inequality follows from \Cref{lem:add_f_approx}.
This completes the proof.
\end{proof}

\Cref{cor:ks_to_online_contract} in conjunction with \Cref{thm:kssmallagents}, immediately implies \Cref{thm:additive}.
\begin{theorem}[\cite{dutting2020prophet}]\label{thm:kssmallagents}
    When $f$ is additive, there exists a 3-CR algorithm for the knapsack problem with budget $B$ where each item has size at most $B/2$.
\end{theorem}

\Cref{thm:wmrf} follows from \Cref{cor:ks_to_online_contract} together with the following claim, which we prove below.
\begin{proposition}\label{prop:ks_ocrs_intersection}
    There exists an online algorithm for the online knapsack problem with budget $2/3$ and matroid constraints which achieves $4\cdot (3+e^{-2}) \approx 12.54$-CR.
\end{proposition}

\newcommand{\optKS}{x^K}
\newcommand{\optMat}{x^M}
\newcommand{\ocrsKS}{O^K}
\newcommand{\ocrsMat}{O^M}
\newcommand{\ocrsInter}{O}

We prove \Cref{prop:ks_ocrs_intersection}, by intersecting two OCRS: for the knapsack problem with additive rewards \cite{jiang2022tight} and for additive rewards and matroid constraints \cite{feldman2016online}, which we denote by $\ocrsKS$ and $\ocrsMat$, respectively.
The guarantees of these OCRS are as follows

\begin{theorem}[\cite{jiang2022tight}]\label{thm:ocrsKS}
    If it holds that
    $\sum_i \sum_{(f_i,c_i)} p_i \cdot (c_i/f_i) \le 2/3$,
    then the probability that agent $i$ is accepted, conditioned on it's being active, is at least $1/\gamma = 1/({3+e^{-2}})$.
\end{theorem}

\begin{theorem}[\cite{feldman2016online}]\label{thm:ocrsMat}
    Let $P$ be the matroid polytope corresponding to a matroid $M$, and let $x \in P/2$.
    Let $R(x) \subseteq[n]$ be a random set where every element appears with probability $x$. Let $\ocrsMat(R(x))$ be the set of agents accepted for a given realization $R(x)$, then this set is feasible and satisfies,
    $$
    \Pr_{R(x)}[i \in \ocrsMat(R(x))] \ge 1/2
    $$
\end{theorem}

Consider the online stochastic knapsack problem with size $B = 2/3$ and WMRF valuation $f$.
Observe that for any realization there is an optimal solution $S$ such that $S$ is an independent set in the matroid $\M$ underlying $f$.
Thus, the following is a relaxation of the prophet benchmark for this problem.

\begin{alignat}{2}
\textrm{Relaxed Prophet = } \text{maximize}\quad   & \sum_{i=1}^n p_i\cdot f_i \cdot x_i                        &\quad& \\
\text{subject to}\quad               
& \sum_{i=1}^n p_i\cdot (c_i/f_i)\cdot x_i \le 2/3 &   & \label{eq:LP_ks_const} \\ 
& \sum_{i \in S} x_i \le \text{rank}_\M(S) &  \forall S \subseteq [n] & \label{eq:LP_mat_const} \\ 
& x_i \ge 0  &      & \\
& x_i \le 1  &      & 
\end{alignat}
Observe that the matroid constraint hold for any realization of active agents, while the knapsack constraint holds only in expectation.
Let $z$ be the optimal solution to the above LP.
Our algorithm operates as follows, given agent $i$ and a realization $(f_i,c_i)$:
\begin{enumerate}
    \item Make agent $i$ available with probability $z(f_i,c_i)/2$.

    \item Accept $i$ if and only if both $\ocrsMat$ and $\ocrsKS$ accept.\footnote{The OCRS $\ocrsKS$ of \cite{jiang2022tight} is defined for a knapsack of size $1$, and it accepts as input the realized weights of every item (agent). To adjust to knapsack of size $2/3$, we just feed it with $(3/2)\cdot (c_i/f_i)$, instead of $c_i/f_i$.}
\end{enumerate}

Observe that for any realization active agents, the vector $z_i/2$ satisfies the matroid constraints by \Cref{eq:LP_mat_const}. Thus, the preconditions of \Cref{thm:ocrsMat} hold and agent $i$ is picked by $\ocrsMat$ with probability at least $1/2$, conditioned on its being active.

To see that \Cref{thm:ocrsKS} holds as well, observe that in our algorithm the probability that agent $i$ is available is
$ \pi_i = p_i \cdot (1/2)\cdot z_i$.
Since every entry in $z$ is between 0 and 1, we clearly have 
$\pi_i \cdot (c_i/f_i) \le 1$. Additionally,
\[
\sum_i \frac{c_i}{f_i}\cdot \pi_i
\le
\frac12 \cdot \sum_i 
p_i \cdot \frac{c_i}{f_i} \cdot z_i 
\le \frac13,
\]
where the last inequality follows since $z$ adheres to the expected knapsack constraint (\ref{eq:LP_ks_const}).

\begin{proof}[Proof of \Cref{prop:ks_ocrs_intersection}]
The expected value of our algorithm (where the expectation is over the randomness of $\ocrsMat$ and $\ocrsKS$ as well):

\begin{eqnarray*}
    \E[ALG] & = & \sum_{i=1}^n p_i \cdot f_i \cdot I[i \text{ is available and accepted by } \ocrsMat \text{ and } \ocrsKS]] \\
& \ge & 
\frac{1}{4\gamma}\sum_{i=1}^n p_i \cdot f_i \cdot z(f_i,c_i)\\
& \ge & 
\frac{\E[OPT]}{4\gamma},
\end{eqnarray*}
where the first inequality follows from theorems \ref{thm:ocrsKS} and \ref{thm:ocrsMat}, and the second follows from linearity of expectation.
\end{proof}

\end{document}